\documentclass[sigconf, nonacm]{acmart}
\usepackage{comment}
\usepackage{wrapfig}

\usepackage[textsize=tiny]{todonotes}

\usepackage{tikz}
\usetikzlibrary{positioning, arrows.meta}
\usetikzlibrary{fit}
\usetikzlibrary{backgrounds}
\usetikzlibrary{calc,decorations.pathreplacing}

\newcommand{\mycircle}[1]{\textcircled{\raisebox{-0.5pt}{#1}}}

\definecolor{carnationpink}{rgb}{1.0, 0.65, 0.79}

\usepackage{algorithmic}
\usepackage[ruled]{algorithm2e}
\usepackage{amsmath}
\usepackage{enumitem}

\definecolor{royalblue}{rgb}{0.25, 0.41, 0.88}

\definecolor{auburn}{rgb}{0.43, 0.21, 0.1}

\definecolor{bronze}{rgb}{0.8, 0.5, 0.2}
\definecolor{byzantium}{rgb}{0.44, 0.16, 0.39}

\definecolor{cadetblue}{rgb}{0.37, 0.62, 0.63}
\definecolor{darkcyan}{rgb}{0.0, 0.55, 0.55}
\definecolor{darkpink}{rgb}{0.91, 0.33, 0.5}
\definecolor{darkscarlet}{rgb}{0.34, 0.01, 0.1}
\definecolor{deepcarmine}{rgb}{0.66, 0.13, 0.24}
\definecolor{deepchestnut}{rgb}{0.73, 0.31, 0.28}
\definecolor{desert}{rgb}{0.76, 0.6, 0.42}

\definecolor{fandango}{rgb}{0.71, 0.2, 0.54}
\definecolor{firebrick}{rgb}{0.7, 0.13, 0.13}
\definecolor{harvardcrimson}{rgb}{0.79, 0.0, 0.09}
\definecolor{lapislazuli}{rgb}{0.15, 0.38, 0.61}
\definecolor{lightpastelpurple}{rgb}{0.69, 0.61, 0.85}
\definecolor{darkpastelpurple}{rgb}{0.59, 0.44, 0.84}
\definecolor{darkraspberry}{rgb}{0.53, 0.15, 0.34}
\definecolor{darkviolet}{rgb}{0.58, 0.0, 0.83}
\definecolor{deepcerise}{rgb}{0.85, 0.2, 0.53}

\definecolor{maroon}{rgb}{0.5, 0.0, 0.0}
\definecolor{modebeige}{rgb}{0.59, 0.44, 0.09}
\definecolor{persimmon}{rgb}{0.93, 0.35, 0.0}
\definecolor{britishracinggreen}{rgb}{0.0, 0.26, 0.15}
\definecolor{darkpastelgreen}{rgb}{0.01, 0.75, 0.24}

\definecolor{chestnut}{rgb}{0.8, 0.36, 0.36}
\definecolor{darkcoral}{rgb}{0.8, 0.36, 0.27}
\definecolor{goldenbrown}{rgb}{0.6, 0.4, 0.08}

\newif\ifcolormode
\colormodefalse % Uncomment this line to disable colors

\ifcolormode
  \newcommand{\mr}[1]{{{\color{persimmon}#1}\largertodo[color=persimmon]{MR1}}}
    \newcommand{\mrno}[1]{{{\color{persimmon}#1}}}
\else
  \newcommand{\mr}[1]{{{#1}}}
    \newcommand{\mrno}[1]{{{#1}}}

\fi

\ifcolormode
  
  \newcommand{\mrra}[1]{{{\color{goldenbrown}#1}\largertodo[color=goldenbrown]{MR 2.a}}}
  \newcommand{\mrrb}[1]{{{\color{blue}#1}\largertodo[color=blue]{MR 2.b}}}
  \newcommand{\mrrc}[1]{{{\color{harvardcrimson}#1}\largertodo[color=harvardcrimson]{MR 2.c}}}
  \newcommand{\mrrd}[1]{{{\color{bronze}#1}\largertodo[color=bronze]{MR 2.d}}}
  \newcommand{\mrre}[1]{{{\color{modebeige}#1}\largertodo[color=modebeige]{MR 2.e}}}
\else
  
  \newcommand{\mrra}[1]{{{#1}}}
  \newcommand{\mrrb}[1]{{{#1}}}
  \newcommand{\mrrc}[1]{{{#1}}}
  \newcommand{\mrrd}[1]{{{#1}}}
  \newcommand{\mrre}[1]{{{#1}}}
\fi

\definecolor{mediumslateblue}{rgb}{0.48, 0.41, 0.93}

\ifcolormode
  \newcommand{\mrrra}[1]{{{\color{darkviolet}#1}\largertodo[color=darkviolet]{MR 3.a}}}
  \newcommand{\mrrrb}[1]{{{\color{mediumslateblue}#1}\largertodo[color=mediumslateblue]{MR 3.b}}}
      \newcommand{\mrrrbno}[1]{{{\color{mediumslateblue}#1}}}
  \newcommand{\mrrrc}[1]{{{\color{deepcerise}#1}\largertodo[color=deepcerise]{MR 3.c}}}
  \newcommand{\mrrrd}[1]{{{\color{darkpastelgreen}#1}\largertodo[color=darkpastelgreen]{MR 3.d}}}
    \newcommand{\mrrrdno}[1]{{{\color{darkpastelgreen}#1}}}
  \newcommand{\mrrre}[1]{{{\color{blue}#1}\largertodo[color=blue]{MR 3.e}}}
\else
  \newcommand{\mrrra}[1]{{{#1}}}
  \newcommand{\mrrrb}[1]{{{#1}}}
      \newcommand{\mrrrbno}[1]{{{#1}}}
  \newcommand{\mrrrc}[1]{{{#1}}}
  \newcommand{\mrrrd}[1]{{{#1}}}
      \newcommand{\mrrrdno}[1]{{{#1}}}
  \newcommand{\mrrre}[1]{{{#1}}}
\fi

\ifcolormode
  \newcommand{\exps}[1]{{{\color{darkpink}#1}\largertodo[color=darkpink]{EXP}}}
    \newcommand{\expsno}[1]{{{\color{darkpink}#1}}}
\else
  \newcommand{\exps}[1]{{{#1}}}
    \newcommand{\expsno}[1]{{{#1}}}
\fi

\ifcolormode
  \newcommand{\other}[1]{{{\color{darkscarlet}#1}}}
\else
  \newcommand{\other}[1]{{{#1}}}
\fi

\usepackage{listings}
\lstdefinestyle{myStyle}{
    belowcaptionskip=0\baselineskip,
    breaklines=true,
    frame=none,
    tabsize=2,
    numbers=left,
    numbersep=5pt,
    xleftmargin=2.5ex,
    basicstyle=\ttfamily\small,
}

\newcommand{\classNP}{\textsf{NP}}
\newcommand{\classP}{\textsf{P}}
\newcommand{\DP}{\textsc{DP}}
\newcommand{\C}{\mathcal{C}}
\newcommand{\T}{\mathcal{T}}

\newcommand{\out}{\mathrm{out}}
\newcommand{\smj}{\mathrm{smj}}
\newcommand{\Cout}{C_{\out}}
\newcommand{\Cmax}{C_{\max}}
\newcommand{\Csmj}{C_{\smj}}
\newcommand{\Ccap}{C_{\mathrm{cap}}}

\newcommand{\DPconv}{\texttt{DPconv}}
\newcommand{\DPccp}{\texttt{DPccp}}
\newcommand{\DPsub}{\texttt{DPsub}}

\newcommand{\eps}{\varepsilon}
\newcommand{\sqeps}{\sqrt\varepsilon}

\usepackage{amsmath}

\newcommand{\sparagraph}[1]{\vspace{1mm}\noindent {\bf #1}}

\usepackage{mathtools}
\usepackage{xparse}

\AtBeginDocument{%
  \providecommand\BibTeX{{%
    \normalfont B\kern-0.5em{\scshape i\kern-0.25em b}\kern-0.8em\TeX}}}

\begin{document}

\title{DPconv: Super-Polynomially Faster Join Ordering}

\author{Mihail Stoian}
\orcid{0000-0002-8843-3374}
\affiliation{%
  \institution{UTN}
  \streetaddress{Ulmenstraße 52i}
  \city{Nuremberg}
  \country{Germany}
  \postcode{90443}
}
\email{mihail.stoian@utn.de}

\author{Andreas Kipf}
\orcid{0000-0003-3463-0564}
\affiliation{%
  \institution{UTN}
  \streetaddress{Ulmenstraße 52i}
  \city{Nuremberg}
  \country{Germany}
  \postcode{90443}
}
\email{andreas.kipf@utn.de}

\renewcommand{\shortauthors}{Stoian, et al.}


\begin{abstract}
    We revisit the join ordering problem in query optimization. The standard exact algorithm, \texttt{DPccp}, has a worst-case running time of $O(3^n)$. This is prohibitively expensive for large queries, which are not that uncommon anymore. We develop a new algorithmic framework based on subset convolution. $\texttt{DPconv}$ achieves a super-polynomial speedup over \texttt{DPccp}, breaking the $O(3^n)$ time-barrier for the first time. We show that the instantiation of our framework for the $\Cmax$ cost function is up to 30x faster than \texttt{DPccp} for large clique queries.
\end{abstract}






\maketitle

\section*{Source Code}
\texttt{https://github.com/utndatasystems/DPconv}

\section{Introduction}\label{sec:introduction}

The query optimizer is the heart of any relational database system. One of the fundamental tasks of the query optimizer is join ordering. The problem is to reorder the joins, so that the query execution time is minimized. To this end, one introduces a cost model that acts as proxy for the actual execution time. Since the costs are directly reflected in the query execution time, optimal or near-optimal join orders are indispensable for the overall performance. However, the problem is inherently \classNP-hard~\cite{ik}. This means that, \other{unless $\classP = \classNP$}, one has to resort to the exponential (exact) algorithm for small queries and to greedy strategies otherwise.

\sparagraph{Motivation \& Research Question.} In a seminal work, Selinger introduces the first dynamic program to (exactly) optimize the ordering problem~\cite{selinger}. The key observation is that the optimal solution $S^*$ for a set of relations~$P$, called the problem, satisfies Bellman's optimality principle~\cite{bellman1957dynamic}, namely that $S^*$ is computed from two disjoint subproblems $P_1$ and $P_2$, with optimal solutions $S_1^*$ and $S_2^*$, respectively. The naive algorithm, \other{$\texttt{DPsize}$}, runs in $O(4^n)$-time, which can be reduced to $O(3^n)$ by a careful traversal of the subsets of a given set, \other{algorithm known as \texttt{DPsub}~\cite{vance_maier, vance_phdthesis}}.

\begin{figure}
    \centering
    \includegraphics{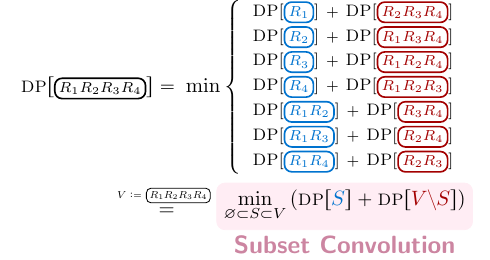}
    \caption{\mrno{How join ordering dynamic programming algorithms, e.g., \texttt{DPsub}, are implicitly using subset convolution. However, they are computing it \emph{naively}. \texttt{DPconv} instead uses a highly-tuned implementation of \emph{fast subset convolution}~\cite{fsc}.}}
    \label{fig:join_ordering_fsc}
\end{figure}

Later, Moerkotte and Neumann~\cite{dpccp} showed that one can obtain an improved algorithm if one disallows cross-products, namely by considering the connectivity structure of the underlying query graph (the algorithm was later extended to hypergraphs~\cite{dphyp}). Their algorithm, \texttt{DPccp}, achieves the lower-bound on the number of connected complement pairs which any dynamic program needs to traverse, as shown by Ono and Lohman~\cite{lohman_cross_products}. Recently, Haffner and \other{Dittrich}~\cite{jo_as_sp} proved that join ordering reduces to computing shortest paths in an exponential-size graph in which the vertices are relation subsets. Their reduction enables the use of well-known speedups via heuristic search, as known from the $A^*$-algorithm. However, while their average-case running time beats that of \texttt{DPccp}, the worst-case running time still remains $O(3^n)$. The $O(3^n)$-time bottleneck leads us to our main question:
\begin{quote}
\centering
\itshape
    Is there a way to break the seemingly unyielding $O(3^n)$-time barrier?
\end{quote}

\mr{
Surprisingly, there is. To this end, consider Fig.~\ref{fig:join_ordering_fsc}, in which we show how the standard join ordering dynamic programming algorithms \texttt{DPsub}~\cite{vance_maier, vance_phdthesis} and \texttt{DPccp}~\cite{dpccp} optimize the full set of relations $V = \{R_1, R_2, R_3, R_4\}$. Simply put, the algorithm iterates over all possible ways to split the original set $V$ into two subsets. \emph{This is exactly a subset convolution.} However, all current join ordering algorithms, \texttt{DPsize}, \texttt{DPsub}, and \texttt{DPccp}, perform it naively, i.e., the expression is evaluated as is. Fortunately for our community, research in algorithm design has led to a \emph{fast} subset convolution~\cite{fsc}. Intuitively, fast subset convolution no longer naively enumerates subsets, but instead uses an FFT-inspired strategy that avoids redundant computational steps of the naive evaluation. To this end, we develop a new exact algorithmic framework based on fast subset convolution that has super-polynomial speedup over $\texttt{DPccp}$~/~$\texttt{DPsub}$. This breaks the long-standing $O(3^n)$ time-barrier for the first time.
}

We instantiate the framework for two well-studied cost functions, $\Cout$ and $\Cmax$, which guarantee time- and space-optimality of query execution, respectively. Namely, the latter minimizes the sum of the intermediate join sizes, while the former minimizes the largest intermediate one.
This results in an $O(2^n n^2 W n \log W n)$-time algorithm for $\Cout$, which is $\widetilde O(2^n)$ when the largest join cardinality $W$ is polynomial in $n$,\footnote{The notation $\widetilde O$ hides poly-logarithmic factors; in this particular case, $n^{O(1)}$.} and an $O(2^n n^3)$-time algorithm for $\Cmax$; the latter running time is independent of $W$.
At a practical level, we show that the instantiation for $\Cmax$ is up to 30x faster than the classic algorithm for clique queries of 17 or more relations.

We further reduce the optimization time for $\Cout$ to $\widetilde O(2^{3n/2} / \sqeps)$-time using an $(1 + \eps)$-approximation algorithm. Unlike our exact algorithm, the running time of this algorithm is independent of $W$.

In addition, we devise a new cost function which combines the benefits of $\Cout$ and $\Cmax$ and provide an implementation which first computes the optimal $\Cmax$ value and then runs a pruned $\Cout$ optimization. We show that this optimization is faster than that of the ``vanilla'' $\Cout$ when using our new framework.

\sparagraph{Contribution.} We summarize our contributions below:

\begin{enumerate}
    \item We introduce a new exact algorithmic framework based on subset convolution which breaks the long-standing time-barrier of $O(3^n)$ for the first time.
    \item We provide a practical instantiation of the framework for $\Cmax$, achieving an $O(2^n n^3)$-time algorithm.
    \item We introduce an $(1 + \varepsilon)$-approximation algorithm for the join ordering problem under $\Cout$ in $\widetilde O(2^{3n/2} / \sqrt{\varepsilon})$-time.
    \item We initiate the joint study of $\Cout$ and $\Cmax$: Minimize the sum of the intermediate join sizes so that the largest one is equal to the optimal $\Cmax$ value.
\end{enumerate}

% Running times.
\mrra{
\sparagraph{Running Times.} Let us first relate the running times for $\Cout$ and $\Cmax$, which seem quite disparate at first glance. They both rely on our highly-tuned implementation of fast subset convolution for dynamic programming (Sec.~\ref{sec:layer_dp}), which runs in $O(2^n n^2)$-time. Thus, we can observe a common $O(2^n n^2)$-time factor in both running times. The difference lies in the implementation of the individual cost functions: Optimizing for $\Cout$ introduces an additional $O(W n \log W n)$-time factor resulting from the application of FFT to sequences of length $W n$ (Sec.~\ref{subsec:c_out}), while $\Cmax$ incurs only an $O(n)$-time overhead (Sec.~\ref{sec:simple_dpconv_cmax}).
}

\mrrb{
% Search Space.
\sparagraph{Search Space.} \DPconv{} imposes no restrictions on the shape of the query graph or join tree. Our framework optimizes arbitrary query graphs---both acyclic and cyclic queries---including cliques, which are considered the worst case of join ordering~\cite{simplification}, \emph{and} bushy join trees, as do other join ordering algorithms such as \DPsub{} and \DPccp{}. This includes optimizing for cross-products in the same running time as for arbitrary query graphs. We discuss this aspect as part of Sec.~\ref{sec:background}. Note that our framework also optimizes query hypergraphs, representing non-inner joins~\cite{dphyp} (discussed in Sec.~\ref{subsec:framework_jo_meets_sc}).
}

\mrrc{
% Other Cost Functions.
\sparagraph{Other Cost Functions.} The literature on join ordering also addresses various cost functions beyond $\Cout$ and $\Cmax$, depending on how a join is executed, e.g., by sort-merge join, hash-join, or nested-loop join~\cite{Moerkotte2006BuildingQC}. We demonstrate that our framework can accommodate the cost function associated with the sort-merge join because it satisfies an additive separability property (see Sec.~\ref{subsec:dpconv_other_cfs}). However, the cost functions associated with hash-joins and nested-loop joins do not enjoy this property and thus cannot be mapped to our framework.
}

% Organization.
\sparagraph{Organization.} The rest of the paper is organized as follows:
First, in Sec.~\ref{sec:background}, we formalize the problem of join ordering and that of \other{fast} subset convolution.
Then, in Sec.~\ref{sec:framework}, we introduce \other{\DPconv{}} along with the novel connection between join ordering and subset convolution. We describe the machinery behind fast subset convolution in Sec.~\ref{sec:FSC}, and in Sec.~\ref{sec:layer_dp} we show how to shave a linear factor from the running time of \emph{any} dynamic programming recursion based on subset convolution (including that of \other{\DPconv{}}). Based on this, we provide in Sec.~\ref{sec:simple_dpconv_cmax} a practical algorithm for $\Cmax$. Then, we outline the approximation algorithm in Sec.~\ref{sec:approx}. We propose $\Ccap$ in Sec.~\ref{sec:ccap}, which we start for the first time the joint study of $\Cout$ and $\Cmax$ with. We outline related work in Sec.~\ref{sec:related_work}, provide a discussion in Sec.~\ref{sec:discussion}, and finally conclude in Sec.~\ref{sec:conclusion}.
\section{Background}\label{sec:background}

In this section, we formalize both problems, namely join ordering and subset convolution.

\subsection{Query Graph}

Let $\mathcal{D} = \{R_1, \ldots, R_n\}$ be a database that contains $n$ relations. A select-project-join query $\mathcal{Q}$ is defined as
\begin{equation}
\mathcal{Q} = \Pi_A(\sigma_P(R_1 \times \dotsc \times R_n)),\label{query_eq} 
\end{equation}
where $P$ is the conjunction of predicates that can be both join predicates, i.e., $R_i.a = R_j.b$, and selection predicates, i.e., $R_i.a = const$, and $A$ is the list of attributes required to appear in the output. The operators $\Pi, \sigma$, and $\times$ are the projection, selection and cross-product operators, respectively, as defined in relational algebra~\cite{codd}.

We can model a query as a \emph{query graph} $Q = (V, E)$, where the vertex set $V$ corresponds to the set of relations $\{R_i\}_{i\in[n]}$ of the query and the edge set $E = \{\{R_u, R_v\}\:|\:R_u, R_v \in V\}$ corresponds to the join predicates (called join edges in the sequel). Intuitively, a query can be evaluated by repeatedly joining two relations and replacing one of them with their join. Another prominent way of executing joins by worst-case optimal joins, which are not necessarily \emph{binary} joins anymore~\cite{wcoj}. In this work, we only concentrate on query optimization of binary joins. In this case, the order in which the joins are performed can be represented by a (binary) \emph{join tree}, where the \other{leaf} nodes are the relations and the inner nodes are the corresponding joins.

\subsection{Cost Function}\label{subsec:cost_functions}

\other{To optimize the join order, one introduces a cost function $\C$ which best models the query execution time. The goal} is to minimize the cost function among all possible join trees. Due to the binary structure of \other{a} join tree, the cost function can be represented as a recursive function along a join tree $\T$, as follows:
\begin{equation}
\C(\T) =
\begin{cases}
    0, & \text{if } \T \text{ is a single relation} \\
    c(T) \otimes \C(\T_1) \otimes \C(\T_2), & \text{if } \T = \T_1 \Join \T_2,
\end{cases}\label{eq:general_c}
\end{equation}
where $c$ is the join cardinality function defined on sets of relations, $T$ is the set of relations spanned by the join tree $\T$,\footnote{Having a separate notation for the join tree and the set of relations it spans will prove useful in the following sections.} and $\T_1$ and $\T_2$ are the left and right join subtrees of $\T$, respectively.

Let us instantiate Eq.~\eqref{eq:general_c} for two cost functions, $\Cout$ and $\Cmax$, which guarantee time-optimality and space-optimality of the query execution, respectively:
\begin{align}
\Cout(T) &= c(T) + \Cout(\T_1) + \Cout(\T_2),\\
\Cmax(T) &= \max\{c(T), \Cmax(\T_1), \Cmax(\T_2)\}.
\label{eq:c_instantiation}
\end{align}
We can observe that the ``$\otimes$'' operator has been substituted by ``$+$'' and ``$\max$'', respectively. \mrrc{We discuss the applicability of our framework to other cost functions in the literature in Sec.~\ref{subsec:dpconv_other_cfs}.}

\subsection{Join Ordering and Dynamic Programming}

% In the following, let us assume for simplicity that cross-products are allowed.
By Bellman's optimality principle~\cite{bellman1957dynamic}, the problem of finding the optimal join tree $\T^*$ amounts to finding the optimal \emph{split} of a set of relations $S$ into two disjoint sets $S_1$ and $S_2$, i.e., $S_1 \cap S_2 = \varnothing$ and $S = S_1 \cup S_2$. Consequently, given a cost function $\C$, the problem can be optimized by the following dynamic programming (DP) recursion, which closely follows the definition of $\C$:
\begin{equation}
\textsc{DP}(S) =
\begin{cases}
    0, & \text{if } |S| = 1\\
    c(S) \otimes \displaystyle\min_{\varnothing \subset T \subset S}(\textsc{DP}(T) \otimes \textsc{DP}(S \setminus T)), & \text{otherwise}.
\end{cases}\label{eq:dp}
\end{equation}
Indeed, this is the idea explored by Selinger and the subsequent work~\cite{selinger, vance_maier, vance_phdthesis, dpccp}. In particular, \texttt{DPccp}~\cite{dpccp} optimizes the recursion by considering only sets of relations that induce a connected subgraph; for clique queries, $\DPsub$ and $\DPccp$ are both exactly the recursion above. \mr{As motivated in Fig.~\ref{fig:join_ordering_fsc}, Eq.~\eqref{eq:dp} is a subset convolution. All previous algorithms evaluate it in the naive way, which takes $O(3^n)$-time. \DPconv{} speeds up its computation by employing fast subset convolution~\cite{fsc}.}

\subsection{Subset Convolution}\label{subsec:back_subset_conv}

Subset convolution is one of the important tools in the field of exact algorithms~\cite{fomin_exp_algos, Cygan2015_chapter}. Its fast counterpart, called Fast Subset Convolution (FSC)~\cite{fsc}, represented a breakthrough in the field by reducing the running time from the straightforward $O(3^n)$ to a non-trivial $O(2^n n^2)$.

\mrrre{
\sparagraph{Dynamic Programming Speedups.} The main application of FSC is the speedup of several dynamic programming recursions of well-known \classNP-hard problems, such as the Steiner tree problem~\cite{dreyfus1971steiner} and min-cost $k$-coloring~\cite{Cygan2015_chapter}. While these problems may seem foreign to our research area, there is a striking similarity between the dynamic programming recursion of these problems and that of the join ordering problem. Indeed, they all use an implicit subset convolution. Through our work, join ordering is now becoming part of this family of problems~\cite{protein_network, chromatic_number_fsc, Cygan2015_chapter, dreyfus1971steiner, ponta2008speeding}.}

While our main result mostly uses FSC in a black-box manner, we present the full machinery behind it in Sec.~\ref{sec:FSC}. Note that for an efficient implementation of the dynamic programs, we will revisit the computation of FSC in Sec.~\ref{sec:layer_dp}, and shave a linear factor for generic FSC-based dynamic programs, as well as several constant factors hidden behind the running time's big-O notation.

\mrrra{
\sparagraph{Key Idea.} Let us first gain an intuition about how subset convolution works at a high level. First, note that the \DP-table is by definition a set function: It maps subsets of relations to their corresponding costs. The usual way to refer to a subset structure is by a subset lattice, in our case of order $n$, since we have $n$ relations. This leads to the following setting: Let $f$ and $g$ be two set functions on the subset lattice of order $n$, their subset convolution in the $(+, \cdot)$ ring is defined for \emph{all} $S \subseteq [n] \vcentcolon= \{1, \ldots, n\}$ by
\begin{equation*}
    h(S) = (f \ast g)(S) = \displaystyle\sum_{T\subseteq S} f(T)g(S \setminus T).
\label{eq:subset_conv_sum_product}
\end{equation*}
Let us first make a few observations: First, the above kind of subset convolution, in the $(+, \cdot)$ ring, is not yet what we exactly need in \DPconv{}. In the next paragraphs, we gradually introduce the toolset to support the subset convolution appearing in Eq.~\eqref{eq:dp}. Second, naively evaluating the above equation for all subsets $S$ takes $O(3^n)$-time. This follows from the fact that for each $S$ we have to iterate over all its ${n \choose |S|}$-many subsets $T$.\footnote{Formally, $\sum_{k = 0}^{n} {n \choose k} 2^k = (1+2)^n = 3^n$.}

\sparagraph{Where The Speedup Comes From.} The faster computation in Bj\"orklund et al.~\cite{fsc} has its roots in a simple observation: One can do a calculation similar to the FFT algorithm~\cite{fft}. The reason: FFT was specifically designed to speed up \emph{sequence} convolutions, bringing down the $O(n^2)$-time of the naive algorithm to a (still unbeatable) $O(n \log n)$-time. Its key insight was to (a) map the original sequence into a Fourier space, (b) perform the convolution in that space as a \emph{point-wise} multiplication -- which takes linear time instead -- and (c) bring the result from the Fourier space back into the original one. The authors take a similar path, resulting in a running time of $O(2^n n^2)$. The perhaps only difference to the original FFT algorithm is \emph{how} the subset functions are mapped to a similar Fourier space where the convolution can be transformed into a point-wise multiplication.
}

\mrrrc{
\sparagraph{The Real Deal: Semi-Rings.} Dynamic programs defined on sets often require the computation to be worked out in \emph{semi-rings}. This is also the case in join ordering for $\Cout$ and $\Cmax$: The well-known $\Cout$ works in the $(\min, +)$ semi-ring, while $\Cmax$ works instead in the $(\min, \max)$ semi-ring. The $(\min, +)$ subset convolution of two set functions $f$ and $g$, $h = f \circ g$, is defined for all $S \subseteq [n]$ as
\begin{equation}
    h(S) = (f \circ g)(S) = \displaystyle\min_{T\subseteq S}\: (f(T) + g(S \setminus T)).
\label{eq:subset_conv_min_plus}
\end{equation}
Unlike the previous kind of subset convolution, computing in the $(\min, +)$ semi-ring results in a different time-complexity landscape. Surprisingly enough, in the general setting, the naive $O(3^n)$-time algorithm which we saw before has the best running time so far. However, in the case where the values of the set functions are bounded integers, one can leverage the previous fast subset convolution for the $(+, \cdot)$ ring. We will come to this in the next section.

\subsection{Rings \& Semi-Rings}\label{subsec:ring_semiring}

We have mentioned rings and semi-rings several times so far. We now want to introduce them in the context of dynamic programming and, more specifically, join ordering. We will focus in particular on the $(+, \cdot)$ ring and the $(\min, +)$ and $(\min, \max)$ semi-rings. Note that we are only aiming for an intuitive understanding of why supporting the latter is much more different than the simple ring setting, where fast subset convolution can work directly in $O(2^n n^2)$ time, as shown before.

\sparagraph{Intuition.} Within the pair of operators of a (semi-)ring, the first one is decisive. In our case, while the $(+, \cdot)$ ring has ``+'' as its first operator, both $(\min, +)$ and $(\max, \max)$ have ``$\min$''. To understand the contrast between these, consider the following illustrative example: If we calculate $2 + 3 + 5 = 10$ and want to remove one of the first terms, e.g., $2$, we can recover the sum of the \emph{other} terms by using the inverse of $2$, i.e., $10 + (-2) = 8$. Things are not so clear in the case of ``$\min$'': If we have $\min\{2, 3, 5\} = 2$ and want to remove $2$, we cannot simply recover the minimum of the remaining elements, $\min\{3, 5\}$. The underlying problem is that has $2$ no inverse. This example may seem artificial at first, but this very problem occurs when applying the inverse map to come back from the Fourier space -- essentially step (c) above. How can this be alleviated? What Bj\"orklund et al.~\cite{fsc} propose is a standard trick in algorithms: embed the $(\min, +)$ semi-ring in the $(+, \cdot)$ ring. We explain and exemplify this technique in Sec.~\ref{subsec:abstract_representation}.}

In the following, we relate the join ordering problem to fast subset convolution for the first time, and provide a unified framework that can be instantiated for several cost functions.

\begin{algorithm}[!t]
  \caption{\texttt{DPconv}: Using fast subset convolution (FSC) to gradually optimize the dynamic programming table.}\label{algo:dp_conv}
  \begin{algorithmic}[1]
  \STATE \textbf{Input:} Query graph $Q = (V, E)$, cardinality function $c$
  \STATE \textbf{Output:} Optimal cost value w.r.t. $\C$
  \STATE $\textsc{DP}[\varnothing] \leftarrow +\infty$
  \STATE $\textsc{DP}[\{R_i\}] \leftarrow 0, \forall R_i \in V$
  \FOR{\textbf{each} $k$ \textbf{in} $2, \ldots, |V|$}
      \STATE $\textsc{DP}' \leftarrow \textsc{FSC}_{(\min, \otimes)}(\textsc{DP}, \textsc{DP})$
      \STATE $\textsc{DP}[S] \leftarrow \textsc{DP}'[S] \otimes c(S), \forall S$ s.t. $|S| = k$
  \ENDFOR
  \RETURN\!$\textsc{DP}[V]$
  \end{algorithmic}
\end{algorithm}

\section{Our Framework}\label{sec:framework}

Let us consider the optimization of an \emph{arbitrary} cost function $\C$ in its associated $(\min, \otimes)$ semi-ring under a generic framework. We will then instantiate the framework for $\Cout$ and $\Cmax$, respectively.

\subsection{Join Ordering Meets Subset Convolution}\label{subsec:framework_jo_meets_sc}

The key observation behind our results is the (now trivial) observation that the definition of DP-recursion, Eq.~\eqref{eq:dp}, is similar to that of subset convolution, Eq.~\eqref{eq:subset_conv_min_plus}. In particular, we show that join ordering falls into the category of dynamic programs which fast subset convolution has already been applied to. In our specific context, there are a few (minor) issues that need to be addressed for FSC to be applicable, issues that have also been considered by Bj\"orklund et al.~\cite{fsc} for other problems, namely:

\begin{enumerate}[label=(\roman*)]
\item The dynamic program $\textsc{DP}$, Eq.~\eqref{eq:dp}, is defined recursively.
\item The subset $T$ of $S$ in the same Eq.~\eqref{eq:dp} must not take $\varnothing$ nor $S$ as value.
\end{enumerate}

\sparagraph{Overview.} Both issues are resolved by a simple technique: We apply FSC layer-wise, i.e., we optimize sets of size 2 first, then those of size 3, and so on -- this is what we call a \emph{layer}. Specifically, at each layer $k$, since the $\textsc{DP}$-table has been computed for layers $k' < k$, we can directly optimize $\textsc{DP}[S]$ for all $S$ with $|S| = k$ by a call to FSC. To alleviate issue (ii), we set $\DP[\varnothing]$, i.e., the DP-cell representing the empty set of relations, to $+\infty$.\footnote{This is not mathematically rigorous. One has to \emph{define} what $+\infty$ in the specific semi-ring is.} Since FSC is called $n$ times, the total optimization time adds up to $O(2^n n^3 \tau_\C)$, where the function $\tau_\C$ is tailored to the specific cost function $\C$ and accounts for the time overhead of semi-ring operations. We will cover the exact expressions of $\tau_\C$ for individual cost functions in the following sections (see Sec.~\ref{subsec:c_out} for $\Cout$ and Sec.~\ref{subsec:c_max} for $\Cmax$, respectively).

Note that we will shave a factor of $O(n)$ from the running time of generic FSC-based dynamic programs, including ours, in Sec.~\ref{sec:layer_dp}, thus reducing the total running time to $O(2^n n^2 \tau_C)$-time for a cost function $\C$.

\sparagraph{Pseudocode.} In Alg.~\ref{algo:dp_conv}, we outline the pseudocode behind our framework \texttt{DPconv}. It takes the query graph $Q$ and the join cardinality function $c$ as input and outputs the optimal cost value w.r.t. the specific cost function $\C$ to which the semi-ring $(\min, \otimes)$ corresponds. It first optimizes the base cases, namely for the empty set of relations and for all sets containing only one relation. The former are initialized with $+\infty$, as argued above, and the latter with 0, cf.~Eq.~\eqref{eq:dp}. Then, at each layer $k$, we optimize the subsets of size $k$ by calling FSC on the current state of the $\textsc{DP}$-table (line 6) and then update the values with the join cardinalities of the subsets (line 7). To this end, note that Alg.~\ref{algo:dp_conv} can optimize for cross-products out-of-the-box: We simply need to also use the cardinalities of all cross-products in $c$. The running time remains naturally the same. Finally, we return the optimal cost.

\begin{algorithm}[!t]
  \caption{\texttt{BuildJoinTree}: Recursively extracting the optimal bushy join tree from the $\textsc{DP}$-table}\label{algo:build_join_tree}
  \begin{algorithmic}[1]
  \STATE \textbf{Input:} Subset of relations $S$, \textsc{DP}-table
  \STATE \textbf{Output:} The optimal bushy join tree
  \STATE \textbf{if} $|S| = 1$ \textbf{return} $S$ \textbf{end if}
  \FOR{\textbf{each} $\varnothing \subset T \subset S$}
    \IF{$c(S) \otimes \textsc{DP}(T) \otimes \textsc{DP}(S \setminus T) = \textsc{DP}(S)$}
        \RETURN\!$(\texttt{BuildJoinTree}(T),\:\texttt{BuildJoinTree}(S \setminus T))$
    \ENDIF
  \ENDFOR
  \end{algorithmic}
\end{algorithm}

\sparagraph{Join Tree Extraction.}\label{par:join_tree_extraction} Note that unlike previous algorithms, Alg.~\ref{algo:dp_conv} does not maintain an \textsc{OPT}-table that stores the optimal split for each subset $S$. This is because \textsc{FSC} itself does not keep track of this information during its execution. In contrast, after the $\textsc{DP}$-table is fully-optimized, we can extract the optimal join tree from the \textsc{DP}-table itself, as outlined in Alg.~\ref{algo:build_join_tree}. Specifically, for each set $S$, we find the subset $T$ that was \emph{intrinsically} used in $\textsc{FSC}$ to optimize $S$, i.e., $\textsc{DP}[S] = \textsc{DP}[T] + \textsc{DP}[S \setminus T]$. Since there are at most $n$ levels of recursion, the worst-case running time for Alg.~\ref{algo:build_join_tree} reads $O(2^n n)$.

\mrrd{
\sparagraph{Cross-Products.} While allowing cross-products can lead to better overall costs~\cite{lohman_cross_products, radke}, the search space increases exponentially~\cite{lohman_cross_products}. A prominent way to deal with cross-products is to heuristically insert them when they are guaranteed to be beneficial~\cite{lohman_cross_products}; this tends to be the case when the estimated cardinality of the input is small enough~\cite{lohman1988grammar_cross_products, radke}.

A natural question is whether \DPconv{} could also support the optimization of cross-products, and whether this particular optimization would take more time than previously specified. Similar to \DPsub{}~\cite{vance_maier, vance_phdthesis}, we can use the cardinalities of the cross-products directly in $c$ (line 7, Alg.~\ref{algo:dp_conv}). This means that \DPconv{} can optimize for cross-products for both cost functions without any overhead.
}

\mrre{
\sparagraph{Query Hypergraphs.} A standard way to model arbitrary non-inner joins in the query graph is to introduce corresponding binary join hyperedges. A binary join hyperedge $h = (A, B)$ connects two sets of relations $A$ and $B$~\cite{dphyp}. This is a generalization of the regular join edge which connects only two relations. In the query hypergraph setting, whenever we want to join two sets of relations connected by a hyperedge, we have to check that both sides are themselves connected \emph{and} there is a hyperedge connecting them. Fortunately, since Eq.~\eqref{eq:dp} enforces that the join cardinalities are taken into account only \emph{after} the DP-layer has been optimized (lines 6-7, Alg.~\ref{algo:dp_conv}), we can directly specify which subgraphs are connected (using Ref.~\cite{dphyp}); this is independent of whether the subgraph contains hyperedges or not. Extending our framework to optimize group-by operators optimally, as in Eich et al.~\cite{eich_groupby}, is an interesting future work.
}

We now come to the embedding technique we motivated and mentioned in Sec.~\ref{subsec:back_subset_conv} that helps us leverage the running time of the fast subset convolution to our employed semi-rings.

\mrrrc{
\subsection{Embedding Technique}\label{subsec:abstract_representation}

Recall the motivating example in the section on rings and semi-rings (Sec.~\ref{subsec:ring_semiring}). To enable the existence of the inverse element, Bj\"orklund et al.~\cite{fsc} propose to embed the semi-ring into a ring,

\sparagraph{Polynomials to the Rescue.} The embedding technique maps the values of the set functions to monomials and then runs the fast subset convolution algorithm in the $(+, \cdot)$ ring. The convolution values can then be read from the resulting polynomials. To see why this works, consider the functions \texttt{[2, 1, 3, 4]} and \texttt{[5, 0, 1, 2]}. When we embed these set functions to monomials, we obtain $\texttt{[}x^\texttt{2}, x^\texttt{1}, x^\texttt{3}, x^\texttt{4}\texttt{]}$ and $\texttt{[}x^\texttt{5}, x^\texttt{0}, x^\texttt{1}, x^\texttt{2}\texttt{]}$, respectively. Thus, by running their subset convolution $\texttt{[}x^\texttt{2}, x^\texttt{1}, x^\texttt{3}, x^\texttt{4}\texttt{]} \ast \texttt{[}x^\texttt{5}, x^\texttt{0}, x^\texttt{1}, x^\texttt{2}\texttt{]}$, we can retrieve the final values as follows: Consider the value at \texttt{001}, which is $x^{\texttt{2+0}} + x^{\texttt{1+5}}$. Note that multiplication between monomials is simply an addition at the exponent level, while the minimum -- in our case, $\min\{2 + 0, 1 + 5\}$ -- is represented by the smallest exponent in the resulting polynomial.

\sparagraph{Representation.} To allow for a seamless instantiation of our framework for other cost functions, we represent the polynomials in \emph{coefficient form}, i.e., pairs of exponents and their associated coefficients. For instance, we represent $2 x + 3x^4$ as $\{(1, 2), (4, 3)\}$.

\sparagraph{Limitation.} The core limitation of the embedding technique is that the size of the coefficient forms exactly corresponds to the largest input value. The reason is that value will be the largest exponent in the entire embedding of the corresponding set function.

In the following, we instantiate the framework for $\Cout$ and $\Cmax$, respectively. In Sec.~\ref{sec:simple_dpconv_cmax}, we show a simpler algorithm to optimize for $\Cmax$ that bypasses the need for the embedding technique.
}

~\\
\subsection{\other{Instantiating} $\Cout$}\label{subsec:c_out}

\other{In the case of $\Cout$, we are working in the $(\min, +)$ semi-ring. To implement the embedding, we simply need to specify how the ``+'' operator should work -- in the most general form, the ``$\otimes$'' operator; compare Eq.~\eqref{eq:general_c}. This corresponds to polynomial multiplication in the coefficient form. Let $P_1$ and $P_2$ be two polynomials in coefficient form. Then $P_1 \otimes P_2$ for an exponent $e$ is defined as}
\begin{equation}
  (P_1 \otimes P_2)(e) = \sum_{\substack{(e_1, c_1) \in P_1\\(e_2, c_2) \in P_2\\e_1 + e_2\:=\:e}} c_1 c_2.
\label{eq:cout_convolution}
\end{equation}
Since the maximum value of the $\Cout$ cost function could be $Wn$ (recall that $W$ is the largest join cardinality) and assuming a FFT-based implementation of the convolution in Eq.~\eqref{eq:cout_convolution}, the factor $\tau_{\out}$ for supporting $\Cout$ is $O(Wn \log Wn)$.

\subsection{\other{Instantiating} $\Cmax$}\label{subsec:c_max}

\other{We now specify the embedding for the $(\min, \max)$ semi-ring. Unlike $\Cout$, we need to specify how the ``$\max$'' operator should work. Namely, the coefficient of exponent $e$ of two polynomials $P_1$ and $P_2$ in coefficient form reads:}
\begin{equation}
  (P_1 \otimes P_2)(e) = \sum_{\substack{(e_1, c_1)\:\in\:P_1\\(e_2, c_2)\:\in\:P_2\\\max(e_1, e_2)\:=\:e}} c_1 c_2. 
\label{eq:cmax_convolution}
\end{equation}
\other{
The intuition is that all exponents \emph{below} $e$ contribute to its final coefficient. If used as in Eq.~\eqref{eq:cmax_convolution}, the size of the coefficient form will still be $W$, as in the case of $\Cout$; this is prohibitively expensive. While there is indeed a way to mitigate this and obtain a running time of $O(2^n n^4)$, which is independent of $W$, we discovered a much simpler algorithm with an even better running time of $O(2^n n^3)$, which does not require the embedding technique. This is understandable due to the fact that, in the $(\min, \max)$ semi-ring, we are not creating \emph{new} values, as is the case in $\Cout$. To not burden the reader with the technicalities of the first approach, we will present directly this algorithm (see Sec.~\ref{sec:simple_dpconv_cmax}).
}

\mrrc{
\subsection{Beyond $\Cout$ and $\Cmax$}\label{subsec:dpconv_other_cfs}

Beside $\Cout$ and $\Cmax$, literature on join ordering also considers other cost functions: Moerkotte~\cite{Moerkotte2006BuildingQC} mentions cost functions related to (a) nested-loop, (b) hash, and (c) sort-merge joins. These cost functions have been mainly designed for left-deep join trees. However, it is an easy exercise to remodel them to work on bushy joins trees. We show that \DPconv{} can be extended to the cost function associated to the sort-merge join.

The problem is that these cost functions require that $c(T)$ be rewritten in terms of $\T_1$ and $\T_2$. That is, instead of $c(T)$, we would need to write $c(T_1, T_2)$; see the below example. The key idea to solve this is to first check whether $c(T_1, T_2)$ can be separated into two independent terms depending only on $T_1$ and $T_2$, respectively.

\sparagraph{When It Works.} We take as running example the sort-merge join cost. Adapting the definition by Moerkotte~\cite[Sec.~3.1.3]{Moerkotte2006BuildingQC}, we have:}
\begin{equation}
\C_{\text{smj}}(\T) =
\begin{cases}
    0, & \text{if } \T \text{ is a single relation} \\
    c(T_1) \log c(T_1)\\\:+\:c(T_2) \log c(T_2)\\\:+ \:\Csmj(\T_1) + \Csmj(\T_2), & \text{if } \T = \T_1 \Join \T_2,
\end{cases}\label{eq:c_smj}
\end{equation}
\mrrc{
where $T_1$ and $T_2$ are the set of relations corresponding to $\T_1$ and $\T_2$, respectively, and $c(T_1)$ and $c(T_2)$ are the corresponding join cardinalities. Note the change to our original $c(T)$ in Eq.~\eqref{eq:general_c}: The split $\T = \T_1 \Join\T_2$ now plays a role. For our framework, this means that we can no longer simply optimize the subset convolution part separately (see line 6 of Alg.~\ref{algo:dp_conv}). We also need to account for the actual sort-merge join cost at the current join, $c(T_1) \log c(T_1) + c(T_2) \log c(T_2)$. However, there is a simple solution to fix this: We can separate the sort-merge join cost and integrate that into the corresponding side, i.e., either $\T_1$ or $\T_2$. Concretely, we need to modify line 6 in Alg.~\ref{algo:dp_conv} as follows:
\[
    \textsc{FSC}_{(\min, +)}(\DP + c \log c),
\]
where the inner function, $\DP + c \log c$, is applied point-wise to each set $S \subseteq [n]$. Put simple, we also add to each DP-entry the sort-merge join cost corresponding to each side. This does not incur a large overhead in the optimization time, as the addition can be performed when pre-processing the zeta transforms of the DP-layers (see Sec.~\ref{subsec:layer_wise_zeta}).

\sparagraph{When It Does Not Work.} Note that this adaptation to the sort-merge join worked because we could split the initial $c(T_1, T_2)$ into two separate cost factors that could be ``sinked'' in the entries of the DP table. To support other cost functions, they need to have a similar \emph{additive} separation property. For instance, the nested-loop join cost, $c(T_1)c(T_2)$, does \emph{not} enjoy this property. This means that our subset convolution based framework, \DPconv{}, cannot be extended to this cost function. A similar situation holds for the hash-join cost, at least in the (classic) setting we are considering: $c(T_1, T_2) = 1.2 \max\{c(T_1), c(T_2)\}$. This extension from Moerkotte's definition mainly designed for left-deep join trees~\cite[Sec.~3.1.3]{Moerkotte2006BuildingQC} takes into account that the hash-table is built on the smaller side, known as the build side. The issue is that ``$\max$'' in $c(T_1, T_2)$ destroys its additive separability into two cost functions depending only on $T_1$ and $T_2$. Therefore, the hash-join cost function cannot also be supported in our framework. 
}
\section{Fast Subset Convolution}\label{sec:FSC}

We next describe Fast Subset Convolution (FSC). We take a closer look at FSC from a practical perspective, so that in Sec.~\ref{sec:layer_dp} we can shave the promised $O(n)$-time factor from the running time of $\texttt{DPconv}$. In the following, we adopt the notation from the Parameterized Algorithms book~\cite{Cygan2015_chapter}, since it has established itself in the literature compared to that of Bj\"orklund et al.~\cite{fsc}.

\subsection{Zeta Transform}

A fundamental operation in FSC is the zeta transform, defined as
\begin{equation}
    (\zeta f)(S) = \displaystyle\sum_{T \subseteq S} f(T),
\label{eq:zeta}
\end{equation}
for any $S \subseteq [n]$. \mrrra{That is, the zeta transform sums $f$ at all subsets of $S$. Naively, this can be computed in $O(3^n)$-time for all $S\subseteq[n]$.} However, we can compute it in $O(2^n n)$-time by observing that we can reuse the computation done for subsets. We detail this in Sec.~\ref{subsec:fsc_impl}.

\subsection{Ranked Convolution}\label{subsec:ranked_convolution}

Given $\zeta f$ and $\zeta g$, the zeta transform of the actual convolution $h = f \ast g$, i.e., $\zeta h$, can now be computed point-wise. To this end, Bj\"orklund et al.~\cite{fsc} employ a \emph{ranked} convolution. Formally,
\begin{equation}
    (\zeta h)(S, r) = \displaystyle\sum_{d = 0}^{r} (\zeta f)(S, d)(\zeta g)(S, r - d),
\label{eq:ranked_convolution}
\end{equation}
for any $S \subseteq [n]$, where $|S| = r$. Thus, we have to apply a zeta transform for \emph{each} rank, i.e., for each cardinality in $\{0, \ldots, n\}$. The ranked convolution can then be computed naively in $O(2^n n^2)$, as for each rank $r$ we need to iterate over all $d \leq r$.

\subsection{M\"obius Transform}\label{subsec:moebius_transform}
To obtain the actual convolution, one applies the M\"obius transform rank-wise. The M\"obius transform is indeed the \emph{inverse} of the zeta transform, i.e., $\zeta \mu = \mu \zeta = \text{id}$, and is defined for any $S \subseteq [n]$ as
\begin{equation}
    (\mu f)(S) = \displaystyle\sum_{T \subseteq S} (-1)^{|T|} f(T).
\label{eq:moebius}
\end{equation}
% We visualize the zeta and M\"obius transforms on a subset lattice of order 3 in Fig.~\ref{fig:zeta_moebius}.
A full-fledged example of FSC is shown in Sec.~\ref{subsec:running_example} and its associated Fig.~\ref{fig:subset_conv_vis}.

% \begin{figure}[h]
%     \centering
%     \includegraphics[width=0.5\textwidth]{revision-figures/_DPconv__Revision_Figures-51.pdf}
%     \caption{Example: The zeta and M\"obius transforms.}
%     \label{fig:zeta_moebius}
% \end{figure}

\subsection{Implementation}\label{subsec:fsc_impl}

\subsubsection{Zeta Transform}

\begin{figure*}
    \centering
    \includegraphics[width=\textwidth]{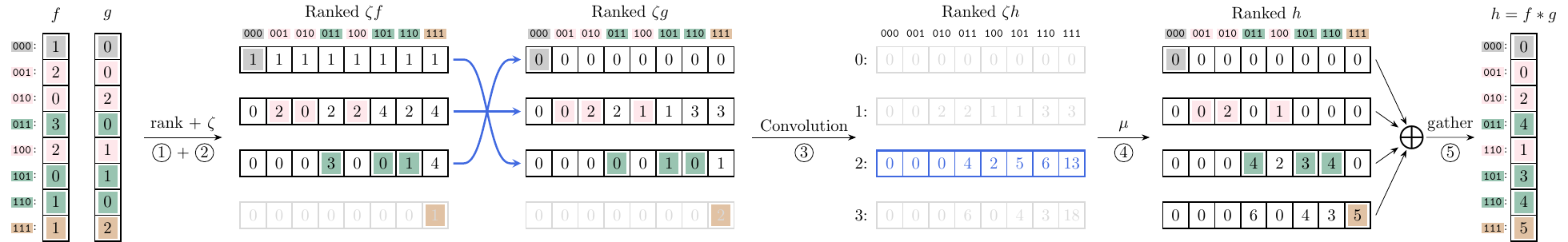}
    \caption{\mrrrbno{Visualizing the fast subset convolution (FSC), outlined in Lst.~\ref{lst:fsc_impl}: \mycircle{1} We rank the set functions $f$ and $g$ and \mycircle{2} apply the zeta transform to obtain $\zeta f$ and $\zeta g$, respectively. \mycircle{3} We perform the ranked convolution between $\zeta f$ and $\zeta g$. \mycircle{4} We apply the M\"obius transform to obtain the ranked $h$. \mycircle{5} Finally, we reconstitute $h = f \ast g$, the actual subset convolution. We highlight in \textcolor{royalblue}{color} the steps needed to compute the second rank ``slice'' of $\zeta h$, namely $(\zeta h)(:, 2)$, during ranked convolution (as in Sec.~\ref{subsec:ranked_convolution}). Intuitively, we need to sum up the dot products between the corresponding slices, i.e., $(\zeta f)(:, 0)$ with $(\zeta g)(:, 2)$, $(\zeta f)(:, 1)$ with $(\zeta g)(:, 1)$, and $(\zeta f)(:, 2)$ with $(\zeta g)(:, 0)$.}}
    \label{fig:subset_conv_vis}
\end{figure*}

A naive evaluation of Eq.~\eqref{eq:zeta} leads to an $O(3^n)$-time algorithm, as for each subset we are to sum up along all its subsets. However, there is a faster way computing it, commonly referred to as Yates' algorithm~\cite{frankyates1937}. Define $\hat f_0(S) = f(S)$ for all $S \subseteq [n]$, and then iterate for all $j = 1, 2, \ldots, n$ and $S \subseteq [n]$ as follows~\cite{fsc}:
\begin{equation}
	\hat f_j(S)=
	\begin{cases}
	\hat f_{j-1}(S) & \text{if $j\not\in S$},\\
	\hat f_{j-1}(S\setminus\{j\})+\hat f_{j-1}(S)  & \text{if
	$j\in S$}.
\end{cases}
\label{eq:zeta_eff}
\end{equation}
By induction, one can show that $\hat f_n(S) = (\zeta f)(S)$ for all $S \subseteq [n]$. The computation of Eq.~\eqref{eq:zeta_eff} takes $O(2^n n)$ operations, as for each subset $S$ we need to iterate over its elements. Lst.~\ref{lst:zeta_impl} shows an implementation of Eq.~\eqref{eq:zeta_eff}.
% Zeta transform
\begin{lstlisting}[
    caption=Zeta transform,
    label={lst:zeta_impl},
    language=C++,
    style=myStyle,
    captionpos=b
]
zeta(f):
  for (d = 0; d != n; ++d):
    for (S = 0; S != 2**n; ++S):
      if S & 2**d:
        f[S] += f[S ^ 2**d]
\end{lstlisting}

% FSC.
\begin{figure}[t]
    \begin{tikzpicture}
        \node[inner sep=0pt, anchor=north west] (code) at (0,0) {
\begin{lstlisting}[
                caption={Fast subset convolution, visualized in Fig.~\ref{fig:subset_conv_vis}.},
                label={lst:fsc_impl},
                language=C++,
                style=myStyle,
                captionpos=b
            ]
FSC(f, g):
  for (r in [0, n]):
    // Rank f and g
    for (S with |S| = r):
      f[S, r] = f[S]
      g[S, r] = g[S]

    // Zeta transform
    zf[S, r] = zeta(f[:, r])
    zg[S, r] = zeta(g[:, r])
    
    // Ranked convolution
    for (d in [0, r])
      zh[:, r] += zf[:, d] * zg[:, r - d]

    // Moebius transform
  	h[:, r] = mu(zh[:, r])

  // Reconstitute h
	for (S in [0, 2**n]):
  	h[S] = h[S, |S|]
\end{lstlisting}
        };

        % Coordinates for the curly brackets and circled numbers.
        \coordinate (A) at ($(code.north east) + (0.5, -1.0)$);
        \coordinate (B) at ($(code.north east) + (0.5, -2.8)$);
        \coordinate (C) at ($(code.north east) + (0.5, -4.3)$);
        \coordinate (D) at ($(code.north east) + (0.5, -5.8)$);
        \coordinate (E) at ($(code.north east) + (0.5, -7.0)$);

        % Curly braces and circled numbers.
        \draw[decorate,decoration={brace,amplitude=5pt},thick] (A) -- ++(0,-1.5) node[midway,xshift=15pt] {\mycircle{1}};
        \draw[decorate,decoration={brace,amplitude=5pt},thick] (B) -- ++(0,-1.25) node[midway,xshift=15pt] {\mycircle{2}};
        \draw[decorate,decoration={brace,amplitude=5pt},thick] (C) -- ++(0,-1.25) node[midway,xshift=15pt] {\mycircle{3}};
        \draw[decorate,decoration={brace,amplitude=5pt},thick] (D) -- ++(0,-1.00) node[midway,xshift=15pt] {\mycircle{4}};
        \draw[decorate,decoration={brace,amplitude=5pt},thick] (E) -- ++(0,-1.25) node[midway,xshift=15pt] {\mycircle{5}};
    \end{tikzpicture}
\end{figure}

\subsubsection{M\"obius Transform} The M\"obius transform can be computed in a similar way. Define $\check f_0(S) = f(S)$ for all $S \subseteq [n]$, and then evaluate the following recursion~\cite{fsc}:
\begin{equation}
	\check f_j(S)=
	\begin{cases}
	\check f_{j-1}(S) & \text{if $j\not\in S$},\\
	-\check f_{j-1}(S\setminus\{j\})+\check f_{j-1}(S)  & \text{if
	$j\in S$}.
\end{cases}
\label{eq:moebius_naive}
\end{equation}
Then one can show that $\check f_n(S) = (\mu f)(S)$ for all $S \subseteq [n]$ and the computation happens in $O(2^n n)$ operations as well.

A sketch of the entire FSC algorithm is shown in Lst.~\ref{lst:fsc_impl}. We use the already-established Python's slicing notation ``:'' to denote an entire axis of the array, in our case, indexed by bitsets corresponding to the actual sets of relations. \mrrrb{We next show a working example.}

\mrrrb{
\subsection{Example}\label{subsec:running_example}

To facilitate the understanding of how fast subset convolution works, we provide a working example in Fig.~\ref{fig:subset_conv_vis}. In particular, we visualize the steps of Lst.~\ref{lst:fsc_impl}. Our example considers two set functions, $f$ and $g$, of size 8, i.e., a subset lattice of size 3. Note that we have combined steps \mycircle{1} and \mycircle{2} in Fig.~\ref{fig:subset_conv_vis} into a single step.

\sparagraph{\mycircle{1} Rank.} In the first step, we create as many rank ``slices'' as there are set cardinalities; in our case, 4 rank slices. Initially, they all contain only the values corresponding to the positions of the same cardinality. For instance, the slice corresponding to rank 2 is initially comprised of the values at positions \texttt{011}, \texttt{101}, and \texttt{110}. Accordingly, these positions and values are displayed in the same color.

\sparagraph{\mycircle{2} Applying Zeta.} Once we created the rank slices, we can now apply the zeta transform to them. Recall its definition in Eq.~\eqref{eq:zeta}: For each set $S$, we sum all the values of $f$ (and analogously for $g$) of at the indices of $S$'s subsets. To show this, consider the same rank slice 2 of $\zeta f$: The value at \texttt{111} -- which is 4 -- is the sum of 3 + 1. Similar in the rank slice 1 of $\zeta g$: The value at $(\zeta g)(\texttt{111}, 1)$ is made up of the non-zero values $g(\texttt{010})$ and $g(\texttt{100})$.

\sparagraph{\mycircle{3} Ranked Convolution.} Once both ranked $\zeta f$ and $\zeta g$ have been computed, we can run the (ranked) convolution between them, as described in Eq.~\ref{eq:ranked_convolution}. We visualize the steps for computing the rank slice 2 of $\zeta h$ in Fig.~\ref{fig:subset_conv_vis}, namely: The colored arrows connecting $(\zeta f)(:, 0)$ with $(\zeta g)(:, 2)$, $(\zeta f)(:, 1)$ with $(\zeta g)(:, 1)$, and $(\zeta f)(:, 2)$ with $(\zeta g)(:, 0)$ show that we need to multiply these rank slices to obtain $(\zeta h)(:, 2)$. As pointed out in Eq.~\eqref{eq:ranked_convolution}, we simply perform a dot product between these and sum up the results. For instance, to obtain $(\zeta h)(\texttt{111}, 2)$, we need to perform the following calculation: $1 \cdot 1 + 4 \cdot 3 + 4 \cdot 0 = 13$.

\sparagraph{\mycircle{4} Applying M\"obius.} To obtain the actual ``ranked'' $h$, we have to apply the M\"obius transform onto ranked $\zeta h$. As explained in Sec.~\ref{subsec:moebius_transform}, the M\"obius transform is the \emph{inverse} of the zeta transform. Once this is done, the next paragraph explains how to obtain the final subset convolution result, $h$. The M\"obius transform is applied as in Eq.~\eqref{eq:moebius}, namely we consider all the subsets of a set $S$ and \emph{subtract} the values where the subset cardinality is odd, and \emph{add} those for even cardinality. For instance, $h(\texttt{111}, 2)$ is computed as follows: The values $\zeta h(\texttt{100}, 2), \zeta h(\texttt{111}, 2)$ are at odd cardinalities, so we subtract their values, while $\zeta h(\texttt{100}, 2), \zeta h(\texttt{101}, 2)$, and $\zeta h(\texttt{110}, 2)$ are at even cardinalities, so we add them. In total, these results in $-4 -13 + 2 + 5 + 6 = 0$, which is exactly $h(\texttt{111}, 2)$.

\sparagraph{\mycircle{5} Gather.} Finally, once the ranked $h$ has been fully computed by applying the M\"obius transform to $\zeta h$, we can obtain the final $h$ by taking a reverse process to step $\mycircle{1}$: Instead of scattering the set functions to rank slices, we now gather the rank slices into one set function. This is done by simply taking the positions from the corresponding rank slice and putting these into $h$; this is also highlighted by the corresponding colors. For instance, to collect the positions $\texttt{011}, \texttt{101},$ and $\texttt{110}$, we take them from the rank slice 2, since all these subsets have cardinality 2.
}
\subsection{Running Time}

Let us calculate the total running time of FSC: The $n$ zeta and M\"obius transforms take in total $O(2^n n^2)$-time, while the rank convolution itself takes $O(2^n n^2)$-time.

When used \emph{as is} in dynamic programming recursions, the running time of FSC is multiplied by a factor $O(n)$. To this end, we show in the next section how to improve the running time from $O(2^n n^3)$ to $O(2^n n^2)$. To motivate this, note that a slowdown of 20x for $n = 20$ in the context of join ordering can make the difference between a practical and an impractical algorithm. 
\section{Layered Dynamic Programming}\label{sec:layer_dp}

Subset convolution is usually employed in definitions of dynamic programs (as our own) where it is called to optimize the $k$th layer of the $\DP$-table (in our case, line 6 in Alg.~\ref{algo:dp_conv}). As previously argued, this call contains a lot of redundancy.

A first observation, $(\star)$, is that, even though we call FSC on the \emph{entire} $\DP$-table (Alg.~\ref{algo:dp_conv}, lines 6-7), we will only update the table for subsets $S$ of size exactly $k$. Taking a look at the internals of FSC explained in Sec.~\ref{sec:FSC}, we observe that the ranked convolution, Eq.~\eqref{eq:ranked_convolution}, actually computes $\hat h(:, r)$ for \emph{each} $r$ in \emph{each} call. This is detrimental, as we will use only the $k$th layer $\hat h(:, k)$ in the $k$th call. A second observation, $(\star\star)$, is that the $\DP$-table itself does not change for subsets of size less than $k$.

In the following, we explain how one can improve the computation of the transforms and that of the ranked convolution given these two observations to reduce the time-complexity of FSC-based DPs from $O(2^n n^3)$ to $O(2^n n^2)$, hence shaving a $O(n)$ factor.

\subsection{Layer-Wise Zeta Transform}\label{subsec:layer_wise_zeta}

With this setting in mind, we can adapt the zeta transforms\footnote{The plural is intended.} intrinsically used in Alg.~\ref{algo:dp_conv} to run faster. Namely, at layer $k > 1$, we do not need to \emph{recompute} the zeta transforms $(\zeta f)(:, j)$, with $j < k$, since due to observation $(\star\star)$, these do \emph{not} change once computed and can, hence, be cached and reused during the entire computation. Consequently, at the $k$th call to FSC, we only need to compute $(\zeta f)(:, k - 1)$.

\subsection{Layer-Wise Ranked Convolution}

In the same manner, by observation $(\star)$, we may skip the outer for-loop (line 2, Lst.~\ref{lst:fsc_impl}) and directly compute $(\zeta h)(:, k)$ once $(\zeta f)(:~,~k~-~1)$ has been computed as previously argued. Notably, due to symmetry -- recall that we actually call FSC with the $\DP$-table -- we may only iterate $d$ until $\lfloor\frac{i}{2}\rfloor$ and multiply $f(:, d)g(:, i - d)$ by~2 (apart from the case when $d$ is indeed equal to $i - d$). Finally, we apply the M\"obius transform on $(\zeta h)(:, k)$ to obtain the $k$th layer of the $\DP$-table (this is symbolically denoted by $\DP'$ in Alg.~\ref{algo:dp_conv}).

Alone these two optimizations shave an $O(n)$-overhead from the running time. We, however, present an additional optimization which, albeit does not reduce the asymptotic time complexity, it does indeed save another constant factor.

\subsection{Avoiding Useless Multiplications}

Recall that our algorithms will work with coefficient forms of polynomials, as described in Sec.~\ref{subsec:abstract_representation}. Hence, the multiplication operator in the ranked convolution (Eq.~\eqref{eq:ranked_convolution} and  Lst.~\ref{lst:fsc_impl}, line 13), is rather expensive since this corresponds to the ``$\otimes$'' operator. We show how to reduce the number of multiplications. This optimization also holds for the simpler algorithm for $\Cmax$ in Sec.~\ref{sec:simple_dpconv_cmax}.

The multiplications take place between \emph{ranked} zeta transforms. Thus, we have $(\zeta f)(S, r) = 0$ for $|S|~<~r$, for any rank $r$. This is because when we apply the zeta transform, we first fill the $r$th layer of the subset lattice with the values of $f(S)$ with $|S| = r$, and then, by construction, only \emph{supersets} will be iterated. Hence, for a rank $r$, sets $S$ with $|S| < r$ will never be touched in the computation of $(\zeta f)(:, r)$.

With this observation, we can prune the range of sets $S$ that we need to consider in line 13 (Lst.~\ref{lst:fsc_impl}) even further. In the following, let $f = g$, as in the context of \texttt{DPconv}. We have:
\begin{align*}
&\forall S.~|S| < d \Rightarrow (\zeta f)(S, d) = 0,\\
&\forall S.~|S| < r - d \Rightarrow (\zeta f)(S, r - d) = 0.
\end{align*}
Consequently, we can simply skip those sets $S$ with $|S|<\max(d, r - d)$. Another further optimization is to restrict ourselves to sets $S$ with $|S| \leq k$. This is because since we only require the $k$th layer, the M\"obius transform only needs to consider sets of maximum cardinality $k$.

% A question that the reader may ask is whether we cannot directly consider sets of cardinality $k$ in line 10 since we anyway output the solution for these sets only. For this, note that in line 10 we are computing $(\zeta h)(:, r)$, and hence we need the results for smaller sets such that the M\"obius transform, Eq.~\ref{eq:moebius}, can work.
\section{A Simple Algorithm for $\Cmax$}\label{sec:simple_dpconv_cmax}

While our framework can support $\Cmax$ (see Sec.~\ref{subsec:c_max}), we found another a much simpler algorithm that does not require the (rather intricate) implementation of the $(\min, \max)$ semi-ring.

\begin{algorithm}
    \DontPrintSemicolon
  \caption{Simpler \texttt{DPconv[max]}: Optimal cost w.r.t. $\Cmax$ in $O(2^n n^3)$-time}\label{algo:dpconv_max_simpler}
  \begin{algorithmic}[1]
  \STATE \textbf{Input:} Query graph $Q = (V, E)$
  \STATE \textbf{Output:} Optimal cost value w.r.t. $\Cmax$
  \STATE $cs \gets \texttt{sort}([c(S) \mid S \subseteq [|V|]], \texttt{decreasing=True})$
  \STATE $p, step \gets 0, 2^{|V| - 1}$
  \WHILE{$step > 0$}
    \STATE $\gamma \gets cs[p + step]$
    \STATE $\textsc{DP} \gets \textsc{LayeredDP}([c \leq \gamma])$ (Sec.~\ref{sec:layer_dp})
    \IF {$\textsc{DP}(V) > 0$}
        \STATE $p \gets p + step$
    \ENDIF
    \STATE $step \gets step\:/\:2$
  \ENDWHILE
  \RETURN\!$cs[p]$
  \end{algorithmic}
\end{algorithm}

\sparagraph{Key Idea.} The key insight is the following: Since we are applying only ``$\min$'' and ``$\max$'' operations, the optimal solution will take its value in the set of join cardinalities. Hence, we can \emph{binary search} the optimal value $\textsc{OPT}$. To check whether a given value $\gamma$ qualifies to be an optimal solution, we apply a technique used by Kosaraju for exact $(\min, \max)$ \emph{sequence} convolution~\cite{kosaraju_min_max}, which we will use on the DP-table itself. The strategy is to first put the DP-entries l.e.q.~$\gamma$ on 1 and those greater than $\gamma$ on 0, and then run FSC, in the $(+, \cdot$) ring, on this modified DP-table. In particular, this refers to one of the layers of the DP-table. We also use our improved layered dynamic programming described in Sec.~\ref{sec:layer_dp}.

\sparagraph{Pseudocode.} We outline the pseudocode of the algorithm in Alg.~\ref{algo:dpconv_max_simpler}. It first \emph{sorts} the join cardinalities in descending order and then performs a binary search on them, searching for the one which separates feasible $\gamma$'s from infeasible ones (note that the maximum join cardinality is always feasible, but may not be the optimum). To this end, we employ Iverson's bracket notation: Given a property $P$, $[P]$ returns 1 is the property is true, 0 otherwise. In our case, $[c \leq \gamma]$ is the following function: \[
    S \mapsto
    \begin{cases}
    1, & \text{if } c(S) \leq \gamma, \\
    0, & \text{otherwise}.
    \end{cases}
\]
Once the $\textsc{DP}$-table has been computed, the algorithm checks whether this value was feasible, i.e., whether $V$ has a positive value in the $\textsc{DP}$-table.
If that is the case, we search for smaller $\gamma$'s.
The algorithm concludes by returning the smallest $\gamma$ for which $\DP(V)$ is still positive.
In the same manner as for the standard $\DPconv$, we can build the join tree once we found the optimal value (see Alg.~\ref{algo:build_join_tree}). \mrrrd{We visualize Alg.~\ref{algo:dpconv_max_simpler} in Fig.~\ref{fig:explainy_figure}.}
\begin{figure}[h]
    \centering
\includegraphics[width=0.65\linewidth]{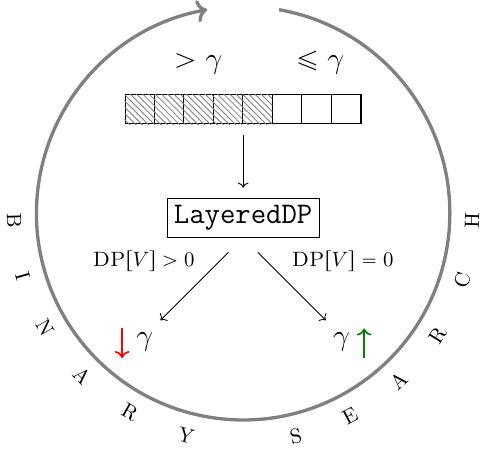}
    \caption{\mrrrdno{Visualizing Alg.~\ref{algo:dpconv_max_simpler}.}}
    \label{fig:explainy_figure}
\end{figure}

\sparagraph{Running Time.} Given our improved implementation of layered dynamic programming (Sec.~\ref{sec:layer_dp}), our new Alg.~\ref{algo:dpconv_max_simpler} runs in time $O(2^n \log 2^n + 2^n n^2 \log 2^n) = O(2^n n^3)$. The additional factor $\log 2^n$ in the second term comes from the running time of the binary search on the $2^n$-sized sorted list of join cardinalities.

\sparagraph{Constant-Factor Optimizations.} While this is sufficient to outperform the standard exact algorithm, there is still an optimization that can be done that only reduces the constant factor hidden in the running time. Namely, for the first layers of the DP-table, we directly hardcode the dynamic programming solution for subsets of cardinality l.e.q.~6. This removes the overhead of subset convolution for these small layers. Note that we still need to compute the zeta transforms of these layers, since they will be used in later layers (as in Lst.~\ref{lst:fsc_impl}, line~13).
\section{Approximation Algorithm}\label{sec:approx}

\sparagraph{Motivation.} A prominent result on approximation algorithms for the join ordering problem is due to Chatterji et al.~\cite{chatterji}, who show that in the case of linear join trees, the problem of approximating the optimal cost $K$ within a factor of $2^{\Theta(\log ^{1-\delta}K)}$ is \classNP-hard, for any $\delta > 0$. We approach the problem from the other end:
\begin{quote}
\centering
    How fast can we approximate the optimal $\Cout$ value within a factor of $(1+\eps)$?
\end{quote}
Indeed, a fast approximation algorithm can enable a faster evaluation of the optimal plan, while incurring a small overhead, specified by the precision parameter $\eps$.

\sparagraph{Our Approximation Algorithm.} To show the benefit of reducing the problem of join ordering to subset convolution, we now show how to obtain an $\widetilde O(2^{3n/2} / \sqeps)$-time approximation algorithm that optimizes $\Cout$ within a multiplicative factor of $(1 + \eps)$. In particular, note that our algorithm is still \emph{exponential}. However, it shows that the $O(3^n)$-time barrier can be overcome when asking about $(1 + \eps)$-approximation algorithms. This is the first result of this kind, which we state in Thm.~\ref{thm:approx_jo}. 

\subsection{Approximate Min-Sum Subset Convolution}

Following a recent result by Bringmann et al.~\cite{approx_min_plus} and the so far unexplored connection between min-plus sequence convolution and min-sum subset convolution, where the latter is the one we reduced join ordering to, Stoian~\cite{stoian_approx} has shown that min-sum subset convolution can be $(1 + \eps)$-approximated in $\widetilde O(2^{3n/2} / \sqeps)$-time; here, $\widetilde O$ hides poly-logarithmic factors in the input size and $\eps$ (note that the input size also consists of the join cardinality function, which is represented as a vector of size $2^n$). For completeness, this is their main theorem:
\begin{theorem}[{\cite[Thm.~3]{stoian_approx}}]
    \label{thm:approx_min_sum_subset_convolution}
    $(1+\eps)$-Approximate min-sum subset convolution can be solved in $\widetilde O(2^\frac{3n}{2} / \sqeps)$-time.
\end{theorem}
This result implied approximation algorithms for several problems that reduce to subset convolution, e.g., the prize-collecting Steiner tree problem~\cite{rehfeldt2022exact}. Thus, by our reduction of the join ordering problem to min-sum subset convolution in Sec.~\ref{sec:framework}, we can obtain an $(1 + \eps)$-approximation algorithm for the join ordering problem as well.

% brings our problem into this family.

\begin{figure}
    \centering
    \includegraphics[width=0.35\textwidth]{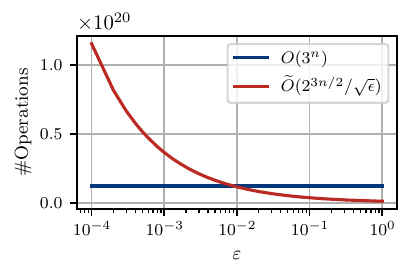}
    \caption{Theoretical number of operations of the exact $O(3^n)$-time algorithm and the $\widetilde O(2^{3n/2} / \sqrt{\varepsilon})$-time $(1+\eps)$-approximation algorithm for $n = 40$ and varying $\varepsilon$'s.}
    \label{fig:approx_plot}
\end{figure}

\subsection{Approximate Join Ordering}

The approximation algorithm follows their simple scheme:

\begin{theorem}
\label{thm:approx_jo}
If $(1 + \eps)$-approximate min-sum subset convolution runs in $T(n, \eps)$-time, then $(1 + \eps)$-approximate join ordering can be solved in $O(T(n, \frac{\eps}{n-1}))$-time.
\end{theorem}
\begin{proof}
Consider the evaluation of the min-sum subset convolution between the $\textsc{DP}$-table and itself at each of the $n - 1$ optimization layers; see Alg.~\ref{algo:dp_conv}, line 6. Fixing $\eps' > 0$ for each convolution call, we obtain a cumulative relative error bounded by $(1 + \eps')^{n - 1}$. By setting $\eps' = \Theta(\frac{\eps}{n - 1})$, we obtain a relative error of at most $\eps$.
\end{proof}
\begin{corollary}$(1 + \eps)$-Approximate join ordering can be solved in $\tilde O(2^\frac{3n}{2} / \sqeps)$-time.
\end{corollary}
In particular, we aim to optimize $\Cout$ (since our $\Cmax$ algorithm in Sec.~\ref{sec:simple_dpconv_cmax} already achieves a better running time). This is particularly interesting since the running time of the approximation algorithm does not depend on $W$, the largest join cardinality. To get an intuition for the running time of the approximation algorithm, we plot in Fig.~\ref{fig:approx_plot} the theoretical number of operations of the exact $O(3^n)$-time algorithm and the $\widetilde O(2^{3n/2} / \sqrt{\varepsilon})$-time $(1 + \varepsilon)$-approximation algorithm for $n = 40$ and varying $\varepsilon$'s. For $\varepsilon = 10^{-2}$, i.e., the optimal value is approximated by a multiplicative factor of $(1 + 10^{-2})$, the runtime of the approximation algorithm outperforms that of the exact algorithm.

Note that the intricate details in the approximation framework by Bringmann et al.~\cite{approx_min_plus} make it hard to have an immediate practical algorithm out of the above theoretical result. We discuss this in Sec.~\ref{sec:discussion}.
\section{Fusing $\Cout$ and $\Cmax$}\label{sec:ccap}

We can indeed regard the faster optimization of $\Cmax$ from another perspective: What if we could optimize the optimal $\Cout$-value under the constraint that the intermediate size is not too large? To show the motivation behind this problem, consider the optimization of $\Cout$ in the case of Q19d in JOB~\cite{job_first_paper}. When using the true cardinalities, the max. intermediate join size of the \emph{optimal} plan w.r.t.~$\Cout$ is 3,036,719 tuples. In contrast, directly optimizing the largest join size via $\Cmax$ only results in an intermediate size of 1,760,645 tuples; this reduces the largest intermediate size by 1.72x. The same can be observed in the recently introduced CEB benchmark: There is a query,\footnote{Specifically, \texttt{11a/5ec72a84a33f3b3b1f4e53b734731ab0bbecebba.sql}} whose optimal $\Cout$ plan has the same behavior. Namely, the largest intermediate join is consists of 11,637,593 tuples, yet if we directly optimized under $\Cmax$, we obtain a largest intermediate join of 9,805,312 tuples.

\subsection{Capping $\Cout$}

Having optimized for $\Cmax$ does not represent any impediment in further refining the plan w.r.t.~$\Cout$. Indeed, we propose a novel cost function to be optimized for, namely the $\Ccap$, motivated by the previous findings. Namely, we propose to \emph{jointly} optimize $\Cout$ and $\Cmax$, i.e., minimize the sum of the intermediate join sizes while enforcing that the largest one is equal to the optimal $\Cmax$ value. This ensures that we both have a bounded intermediate size (space-optimality) \emph{and} the best time-optimal plan under this constraint.

The drawback is naturally that this joint optimization now needs two optimizer passes: (i) Find the optimal $\Cmax$ value, and (ii) optimize $\Cout$ so that all intermediate join sizes are bounded above by that value. To reduce the optimization time of the second pass, we can reduce the search space of the optimization problem, by observing that in $\texttt{DPccp}$ (and $\texttt{DPsub}$) we can directly prune the intermediate solutions the size of which exceed the optimal $\Cmax$ value.

We visualize this preliminary overhead in Fig.~\ref{fig:ccap_overhead}. Note that reducing this overhead is the motivation behind our novel framework, which achieves strongly-polynomial speed-up over standard join ordering algorithm, $\DPccp$~\cite{dpccp}. In Fig.~\ref{fig:ccap_overhead}, we show the price we have to pay for this joint optimization. We optimize the queries of the JOB~\cite{job_first_paper} and CEB~\cite{ceb} benchmarks, respectively, via \texttt{DPccp} as follows: For $\Cout$, this is the classic scenario. For $\Ccap$, we first optimize $\Cmax$ via $\texttt{DPccp}$ and then run $\texttt{DPccp}$ again, optimizing $\Cout$ under the constraint that any intermediate join size is l.e.q.~the previously computed $\Cmax$ value. In the case of JOB, for the largest join queries of 17 relations, the overhead is of ~10ms. This is still negligible, but as we will show in Sec.~\ref{subsec:eval_ccap}, for larger join clique queries the overhead tends to be over 22\%.

\begin{figure}
    \hspace{-1.25em}
    % \centering
    \includegraphics[width=0.49\textwidth]{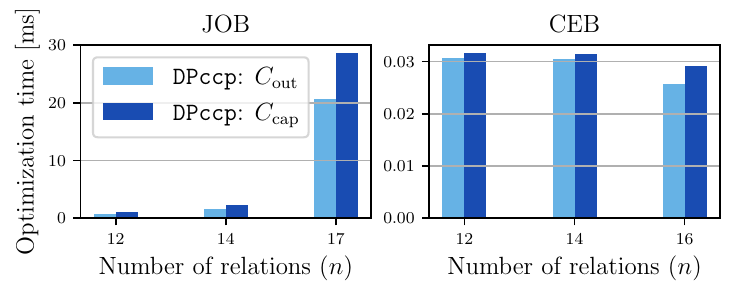}
    \caption{Overhead in optimization time for $\Ccap$ on JOB~\cite{job_first_paper} and CEB~\cite{ceb}, i.e., optimizing $\Cout$ under the constraint that the largest intermediate size is the same as when optimizing with $\Cmax$ (two optimization phases).}
    \label{fig:ccap_overhead}
\end{figure}

\subsection{Reducing Optimization Time}

The optimization of $\Ccap$ has in itself, first, the optimization of $\Cmax$, and then a \emph{pruned} $\Cout$ optimization. If using the standard exact join ordering algorithm, $\texttt{DPccp}$, the running time of the $\Cmax$ optimization is still $O(3^n)$. As discussed in Sec.~\ref{sec:simple_dpconv_cmax}, we can reduce this running time to $O(2^n n^3)$. Therefore, we can simply use $\texttt{DPconv[max]}$ and reduce the optimization time. The second pass, that of optimizing $\Cout$ under the constraint that the largest intermediate size does not exceed this value, remains as before. The advantage is that, since both the first pass is sped up via $\texttt{DPconv[max]}$ and the second pass has a pruned search space, we show that we are even faster than a ``vanilla'' $\Cout$ optimization. We show the corresponding experiments in Sec.~\ref{subsec:eval_ccap}.
\section{Evaluation}\label{sec:evaluation}

We show by means of experiments that $\texttt{DPconv}$ achieves a significant speedup over the standard $O(3^n)$-time join ordering algorithm.

\sparagraph{Experimental Setup.} We perform our experiments on a \texttt{c5.xlarge} EC2 instance which has an Intel Xeon Platinum 8275CL processor with 4 vCPUs and 8 GB of memory. All join ordering algorithms are implemented in C++.

\sparagraph{Benchmark Sets.} We use the setup from the CEB benchmark~\cite{ceb}, which already provides the true cardinalities for IMDb for their 13,644 queries and the 113 queries of JOB~\cite{job_first_paper} (note that this setup has already been used for Fig.~\ref{fig:ccap_overhead}). For clique queries, we generate random join cardinalities $\leq$ 100M, with the constraint that $c(S) \leq c(S_1)c(S_2), \forall S_1, S_2, \subsetneq S, S_1 \cap S_2 = \varnothing, S_1 \cup S_2 = S$, i.e., we do not exceed the cardinality of the cross-product of any possible combination of subset pairs. \mrrd{Note that since we directly optimize on clique queries, the running times can also be considered as that of optimizing for cross-products, as discussed in Sec.~\ref{subsec:framework_jo_meets_sc}.} \mrrb{Moreover, since subset convolution does not (yet) exploit sparse set functions---in our case, corresponding to unconnected query subgraphs---the running time is thus independent of the cyclicity of the query graph; we provide a discussion of this in Sec.~\ref{sec:discussion}.}

\exps{
Whenever we compare to the $\texttt{A}^\texttt{*}$-based algorithm by Haffner and Dittrich~\cite{jo_as_sp}, we use their benchmark set. We use the same evaluation scripts,\footnote{Their reproducibility experiment is available at: \url{https://gitlab.cs.uni-saarland.de/bigdata/mutable/evaluation}} i.e., we use their generated cliques and cardinalities. We show the optimization times for cliques of up to 18 relations in Fig.~\ref{fig:dpconv_vs_astar}.\footnote{In the current version, the evaluation script starts to timeout after 19 relations; we increased the default timeout to 800s, yet this did not solve the issue.}
}

\sparagraph{Competitors.} The standard exact algorithms, $\DPccp$ and $\DPsub$, follow the implementation in the reproducibility experiment of Neumann and Radke~\cite{adaptive}.\footnote{https://db.in.tum.de/$\sim$radke/papers/hugejoins-reproducibility.pdf} We also implement the bitsets as 64-bit integers, which we wrap with helper functions to provide iterators of subsets. $\DPconv$ uses all the optimizations described in Sec.~\ref{sec:layer_dp} for layered dynamic programs. The optimization time includes the time for extracting the join tree from the layered dynamic programming.

\subsection{Super-Polynomial Speedup}

\begin{figure}[!t]
    \centering
    \includegraphics[width=0.45\textwidth]{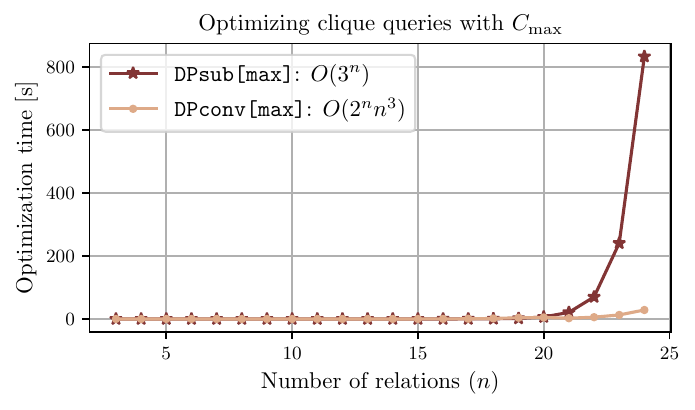}
    \caption{Clique queries optimization: Both \texttt{DPsub[max]} and \texttt{DPconv[max]} optimize for $\Cmax$.}
    \label{fig:cmax_clique}
\end{figure}

Within our framework, $\DPconv$, we have shown that join ordering can be done faster than $O(3^n)$. Specifically, we provided an $O(2^n n^2 Wn \log Wn)$-time algorithm for optimizing $\Cout$, which is $\widetilde O(2^n)$ when the largest join cardinality $W$ is polynomial in $n$, and an $O(2^n n^3)$-time algorithm for optimizing $\Cmax$; this is the first super-polynomial speedup for the join ordering problem. While the algorithm for $\Cout$ is not a practical one, we devised in Sec.~\ref{sec:simple_dpconv_cmax} a simple and practical algorithm for $\Cmax$.

\sparagraph{\DPconv{} vs. \DPsub{}}. We benchmark on clique queries, as these are the hardest queries to optimize for~\cite{simplification}. In particular, $\DPsub$ excels at this type of queries since $\DPccp$ has the overhead of exploring the graph itself (note that this is also the case in the experiments of the original paper~\cite{dpccp}). We show the optimization times for cliques of up to 24 relations in Fig.~\ref{fig:cmax_clique}. The optimization time is averaged for each $n \in \{3, \ldots, 24\}$ across 5 randomly generated instances.

The first observation is that our new algorithm is indeed practical: It starts being faster than $\DPsub$ after 17 relations, and for a large join query of 24 relations, it has a speedup of 29x. Note that $n = 17$ is still in the regime of the JOB benchmark. However, JOB has sparse query graphs, hence $\DPccp$ is enough for such queries, as already shown in Fig.~\ref{fig:ccap_overhead}.

\exps{
\sparagraph{\DPconv{} vs. $\texttt{A}^\texttt{*}$.} Haffner and Dittrich~\cite{jo_as_sp}, as part of their mu\emph{t}able database system, have recently shown that, indeed, using an $A^*$-based optimizer, one can reduce the number of $\texttt{ccp}$ to be explored. Note that this indeed leads to an optimal solution. In particular, unlike $\texttt{DPconv[max]}$, the practical instantiation of \DPconv{} for $\Cmax$, their algorithm optimizes for $\Cout$. Thus, the following experiment solely serves to compare the running times, as the cost functions to be optimized are different.
}
\begin{figure}
    \centering
    \includegraphics[width=0.95\linewidth]{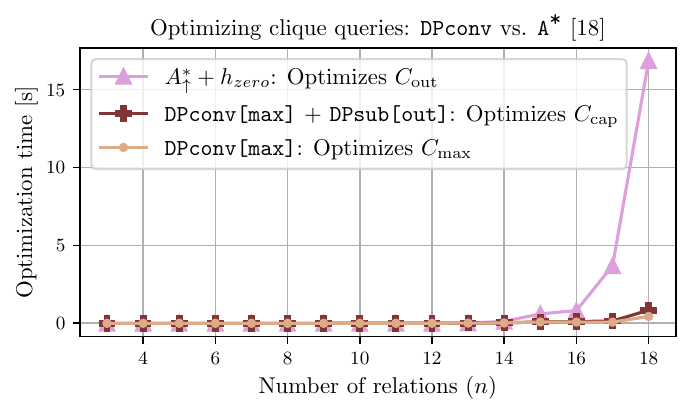}
    \caption{\expsno{Clique queries optimization (setup as in Ref.~\cite{jo_as_sp}): The $\texttt{A}^\texttt{*}$-based optimizer by Haffner and Dittrich~\cite{jo_as_sp} optimizes for $\Cout$, while \texttt{DPconv[max]} optimizes $\Cmax$ in $O(2^n n^3)$-time, and the joint combination between \texttt{DPconv[max]} and the pruned \texttt{DPsub[out]} optimizes for $\Ccap$ (Sec.~\ref{sec:ccap}).}}
    \label{fig:dpconv_vs_astar}
\end{figure}

\subsection{Optimizing $\Ccap$}\label{subsec:eval_ccap}

\begin{figure}
    \centering
    \includegraphics[width=0.48\textwidth]{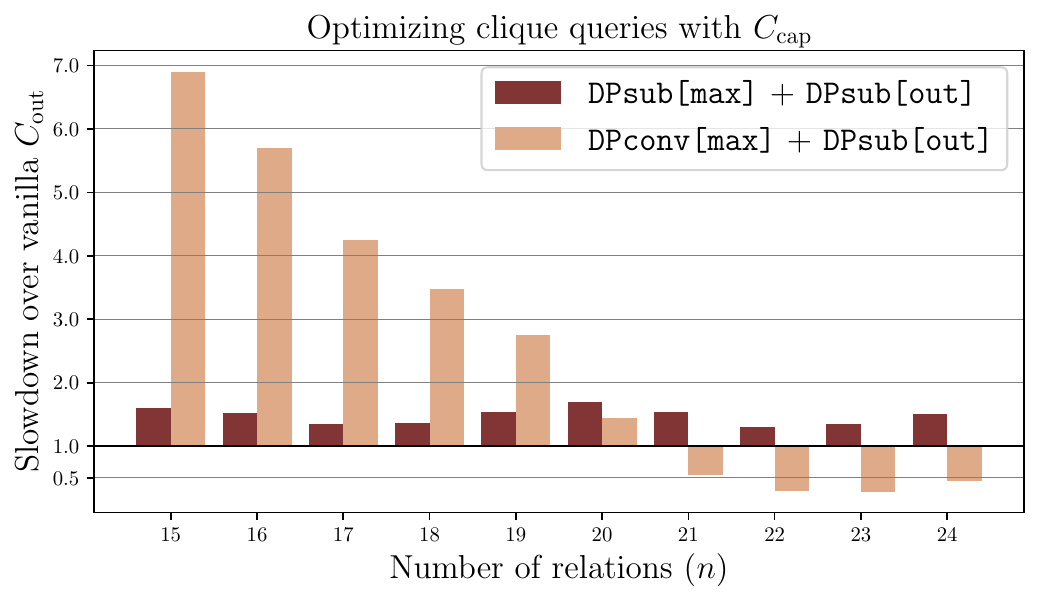}
    \caption{The slowdown of optimizing $\Ccap$ for large clique queries over a ``vanilla'' $\Cout$ optimization. \other{The baseline is $\texttt{DPsub[out]}$.} While a naive optimization is (naturally) slower, using $\texttt{DPconv[max]}$, the instantiation of our novel framework for $\Cmax$, in the first optimization pass, and followed by the pruned $\Cout$ optimization, we obtain an optimization time even faster than that of a ``vanilla'' $\Cout$.}
    \label{fig:ccap_clique_optimizaiton}
\end{figure}

We show that the optimization time of $\Ccap$ can be made practical using our novel $\DPconv$ framework. Recall that optimizing $\Ccap$ requires two optimization passes. We will focus on the first pass, in which we optimize for $\Cmax$. The reason is that we can use $\texttt{DPconv[max]}$, the instantiation of our novel framework $\DPconv$ for the $\Cmax$ cost function. This reduces the running time of this pass from $O(3^n)$-time to $O(2^n n^3)$-time. This is particularly significant for large join queries.

To this end, in Fig.~\ref{fig:ccap_clique_optimizaiton}, we show the slowdown of optimizing $\Ccap$ compared to a ``vanilla'' $\Cout$ optimization. We benchmark on clique queries, as previously argued. To not clutter the plot, we only keep $\DPsub$ as the baseline for clique queries. Thus, the baseline is the optimization of $\Ccap$ via $\DPsub$, namely we first optimize $\Cmax$ and then run a pruned $\Cout$ optimization, i.e., we then skip the subsets whose intermediate size is larger than this latter value. Our proposed algorithm replaces $\texttt{DPsub[max]}$ with $\texttt{DPconv[max]}$ in the first pass. We first observe that, naturally, the naïve $\Ccap$ optimization is slower than the ``vanilla'' $\Cout$ optimization (slow-down is over 22\%). In contrast, having both a super-polynomial speedup for the first optimization due to $\DPconv$ and a pruned search space for the second pass, we are even faster than the ``vanilla'' $\Cout$ optimization after 21 relations. Compared to the $\texttt{A}^\texttt{*}$-based algorithm, the optimization of $\Ccap$ outperforms that of $\Cout$ after 14 relations as well.

\sparagraph{Analyzing $\Ccap$ on CEB.} Out of the 13,644 queries of the CEB~\cite{ceb} benchmark, there are 2,873 queries for which the largest intermediate size in the optimal $\Cout$ plan is 6.8\% larger than the optimal $\Cmax$ intermediate size. For these queries, $\Cmax$ looses 22.8\% in the optimal $\Cout$ value, while $\Ccap$ naturally reduces this to only 9.5\%.
\section{Related Work}\label{sec:related_work}

The literature on join ordering is extensive. This is partly because of the effect that a bad join order can have on the query performance and hence the natural desire to avoid such cases.
As a result, there are a few \emph{exact} algorithms, a small number of \emph{polynomial-time} algorithms for restrictive cases, several \emph{greedy} (non-optimal) algorithms, and a handful of optimizers based on general-purpose solvers.
Our work falls into the category of exact algorithms. In particular, no previous work has observed the link to subset convolution, neither did it achieve a running time as we propose.
We are the first to break the $O(3^n)$ \other{time-barrier} for the join ordering problem on generic query graphs (and bushy solutions).
We divide the related work into exact, approximation, and best-effort algorithms. The latter are either polynomial-time algorithms for special instances or greedy algorithms without any approximation guarantee.

\subsection{Exact Algorithms}

The history of the join ordering problem starts at Selinger, proposing an $O(4^n)$-time algorithm, commonly referred to as $\texttt{DPsize}$~\cite{selinger}.
To some extent, this algorithm \emph{does} subset convolution in the naive way, i.e., it iterates all subsets $T$ of a given set $S$ of relations, but does not do that in time $2^{|S|}$, but rather in time $2^{n}$.
Vance and Maier~\cite{vance_maier} observed this limitation and fixed it within the $\texttt{DPsub}$ algorithm, which takes time $O(3^n)$.
Since $O(3^n)$ seemed rather rigid, not being adaptive to the graph topology, Ono and Lohman~\cite{lohman_cross_products} analyzed the \emph{minimum} number of subplan pairs that have to be iterated in any dynamic program.
To this end, Moerkotte and Neumann~\cite{dpccp} designed $\texttt{DPccp}$, which emulates to the graph topology and obtains as time-bound exactly the number of connected complement pairs (\#\texttt{ccp}'s).
However, the running time $O(3^n)$ still persisted. In their recent work, Haffner and \other{Dittrich}~\cite{jo_as_sp} showed that using the $\texttt{A}^\texttt{*}$ algorithm, one can obtain an algorithm which still outputs the optimal plan without having to explore all \#\texttt{ccp}'s.
This is indeed a promising result, as it \other{shows} that the lower-bound of $\#\texttt{ccp}$ can in some cases be by-passed. However, in the worst case, the running time is still the unyielding $O(3^n)$.
In our work, we obtain for the first time an $\widetilde O(2^n W)$-time algorithm, completely breaking the $O(3^n)$ time-barrier when \other{$W$ is polynomial in $n$}. In the case of $\Cmax$, i.e., minimizing the maximum intermediate join cardinality, we obtain an $O(2^n n^3)$-time algorithm, which is also practical.

\sparagraph{Bottom-Up vs. Top-Down.} It is well known that dynamic programs have two implementations, \emph{bottom-up} and \emph{top-down}, each with its advantages and disadvantages.
One of the most compelling advantages of top-down enumeration is the possiblity of easily integrating cost-bounds so that the search space may be easily pruned~\cite{pruning_fender}.
Hence, Chaudhuri et al.~\cite{chaudhuri} explore the possiblity of implementing join ordering as a top-down procedure, only considering linear solutions.
Building on this work, DeHaan and Tompa~\cite{dehaan} extend the top-down method to bushy join trees, disallowing cross products.
Fender and Moerkotte~\cite{fender11, fender12} improve the running time of these algorithms and get rid of the connectedness check, i.e., only outputting the \texttt{ccp}'s.

\subsection{Approximation Algorithms}

Exact algorithms are rather expensive. To this end, Chatterji et al.~\cite{chatterji} analyzed whether there are instances that can be solved by approximation algorithms in polynomial time. Unless $\classP = \classNP$, the answer remains negative. Specifically, they showed that, for any $\delta > 0$, the problem of approximating the optimal cost $K$ within a factor of $2^{\Theta(\log ^{1-\delta}K)}$ is \classNP-hard (note that this is a poly-logarithmic approximation).

\subsection{Best-Effort Algorithms}

The \classNP-hardness of a fundamental problem is a bitter truth. Therefore, research has focused on finding polynomial-time algorithms for special instances or greedy algorithms for arbitrary query graphs.

\sparagraph{Polynomial-Time Algorithms.} Exponential-time algorithms fail to optimize larger queries in a reasonable time.
To this end, it is interesting to ask which instances admit \emph{polynomial-time} algorithms. The most notable one is the cubic-time algorithm for chain queries. Another class is that of tree queries, for which the \texttt{IKKBZ} algorithms returns the optimal left-deep join tree~\cite{ik,kbz}.
Neumann and Radke~\cite{adaptive} observed that one can use \texttt{IKKBZ} as a sub-routine: They \emph{linearize} the query graph via \texttt{IKKBZ} (since a left-deep solution is inherently a linear ordering of the underlying graph) and then run the cubic-time dynamic program on top to build a near-optimal solution. This strategy yields excellent costs for tree queries. 

\sparagraph{Greedy Algorithms.} Research has also focused on greedy algorithms which can at least \emph{avoid} \other{the} bad plans. The most representative is the Greedy Operator Ordering (\other{\texttt{GOO}})~\cite{goo} that chooses the cheapest sub-plan at each step. This runs in $O(n \log n)$-time, yet it does not come with any optimality guarantee on the output join order.
This gap between exponential-time exact algorithms and purely greedy ones has remained unexplored until Kossman and Stocker~\cite{idp} introduced Iterative Dynamic Programming \other{(\texttt{IDP})} \other{which refines the greedy join orders of large queries. The key insight is to iteratively run exact DP on join subtrees of size $k$.}

\sparagraph{General-Purpose Solvers.} Join ordering has also been approached by several general-purpose solvers, such as genetic algorithms~\cite{steinbrunn_jo}, mixed-integer linear programming~\cite{milp_trummer}, and simulated annealing~\cite{steinbrunn_jo}. Note that these works only approximate the optimal solution (\emph{without} any approximation guarantee).
The problem can also be optimized on quantum hardware via quantum annealing~\cite{quantum_jo_1, quantum_jo_2}.
However, this does not lower the \emph{classical} time-complexity of exact join ordering.
Motivated by the promise of workload-aware query optimization, research also has focused on \emph{learned} alternatives: Marcus and Papaemmanouil~\cite{marcus_jo} suggest using Reinforcement Learning and introduce an agent that outputs the join order and is penalized based on the corresponding join cost. \other{Motivated by the repetitiveness of the queries in cloud workloads~\cite{predicate_caching}, a further promising direction is query super-optimization~\cite{query_superoptimization}.}
\section{Discussion}\label{sec:discussion}

\sparagraph{Resource-Aware Query Optimization.} The trend nowadays is to execute queries in multi-tenant cloud machines. Recently, Viswanathan et al.~\cite{microsoft_container_query_optimization} made the case for \emph{resource-aware} query optimization. The $\Cmax$ cost function can serve as a proxy for the maximum memory consumption of a given query. Minimizing $\Cmax$ of concurrently running queries can help reduce memory spikes.

\sparagraph{Co-Optimizing $\Cout$ and $\Cmax$.} The optimization of $\Cout$ and $\Cmax$ can go beyond our proposed cost function $\Ccap$. With $\Ccap$, we first compute the optimal value of $\Cmax$ and then do a pruned $\Cout$ optimization. Instead of taking the \emph{optimal} $\Cmax$ value, capping $\Cout$ at the 90th percentile of the largest intermediate size allows for more flexibility. So one can effectively trade off between query runtime and memory consumption. This is particularly interesting in cloud scenarios.

The cloud data warehouse Amazon Redshift uses predicted query memory to make scheduling decisions~\cite{wlm}. Instead, one could follow a proactive approach in which a query’s runtime and memory consumption is co-optimized with query scheduling. For example, when there is a high (concurrent) memory load on the system, one would want to minimize the peak memory consumption of newly arriving queries, while when there is low memory load, one can afford a higher memory consumption. Likewise, if there are long-running queries with a low memory footprint in the system, one might want to produce a high memory but fast-running query.

\sparagraph{Practical Implementations.} While we break the $O(3^n)$ time-barrier in the theoretical sense and indeed also provide a practical implementation for $\Cmax$ running in $O(2^n n^3)$-time, it is interesting to further explore practical implementations for $\Cout$, both for the exact (Sec.~\ref{subsec:c_out}) and the approximation algorithm (Sec.~\ref{sec:approx}). In particular, the details of the framework by Bringmann et al.~\cite{approx_min_plus}, upon which the approximate min-sum subset convolution algorithm is based on, span several pages.

\sparagraph{Sparse Subset Convolution.} Subset convolution does not (yet) have a \emph{sparse} counterpart, as is the case for sequence convolution (we refer the reader to Jin and Xu~\cite{jin2024shaving} for the latest results on sparse sequence convolution). This would be particularly useful for sparse query graphs of the JOB~\cite{job_first_paper} and CEB~\cite{ceb} benchmarks. These queries do not benefit from the speedup obtained by $\DPconv$ due to the fact they \other{only} touch at most 17 relations. \other{An algorithmic advance in subset convolution for the sparse setting can be directly transferred to the join ordering problem.}
\section{Conclusion}\label{sec:conclusion}

\other{Join ordering, or finding the optimal order of the joins in a query, is an indispensable task in a database management system. The problem has its roots in the seminal work of Selinger~\cite{selinger}, culminating with the graph-theoretic exact algorithm by Moerkotte and Neumann~\cite{dpccp}. Despite recent research~\cite{jo_as_sp}, the worst-case running time still remains $O(3^n)$.}

In this work, we provided the first super-polynomial speedup over the standard dynamic programming solution. Our framework optimizes (i) $\Cout$ in $\widetilde O(2^n)$-time, when \other{the largest join cardinality $W$ is polynomial in $n$}, and (ii) $\Cmax$ in $O(2^n n^3)$-time. $\DPconv$ is based on subset convolution, a fundamental tool in parameterized algorithms~\cite{Cygan2015_chapter}, and uses the fact that join ordering is implicitly a dynamic programming recursion using subset convolution similar to other classic problems in the literature (see Bj\"orklund et al.~\cite{fsc}). The reduction to subset convolution also implies an $(1 + \varepsilon)$-approximation algorithm for optimizing $\Cout$ in $\widetilde O(2^{3n/2} / \sqrt \varepsilon)$-time.

Beyond the theoretical results, we have made $\DPconv$ practical for database systems. In particular, our algorithm for optimizing $\Cmax$ outperforms the standard exact algorithm for cliques \other{with 17 relations and more}. In addition, we showed that joint optimization of $\Cout$ and $\Cmax$ results in faster optimization times than a ``vanilla'' $\Cout$ after 21 relations, while only increasing $\Cout$ by 9.5\%.

We expect future work on sparse subset convolution to further speed up our framework for query graphs with few connected subgraphs.

\bibliographystyle{ACM-Reference-Format}
\bibliography{dpconv}

\end{document}